\documentclass[a4paper,cleveref]{lipics-v2021}

\nolinenumbers

\newenvironment{socgtabular}{\begin{tabular}}{\end{tabular}}

\hideLIPIcs
\newcommand{\arxivonly}[1]{#1}
\newcommand{\arxivappendix}[1]{#1}
\newcommand{\submissiononly}[1]{}

\usepackage{todonotes}
\usepackage{graphicx}
\usepackage{complexity}
\usepackage{booktabs}
\usepackage{xspace}
\usepackage{caption}
\usepackage{tikz}
\usepackage[export]{adjustbox}

\graphicspath{{figures/}}

\newcommand{\defn}[1]{\textbf{#1}}

\newcommand{\cC}{\mathcal{C}}
\renewcommand{\c}{\gamma}

\renewcommand{\R}{\mathbb{R}}

\newcommand{\si}[1]{\text{SI}(#1)}
\newtheorem{fact}[theorem]{Fact}
\newcommand{\plane}{\overline{\R^2}}

\usetikzlibrary{positioning}
\usetikzlibrary{arrows,decorations.markings}
\usetikzlibrary{shapes}
\usetikzlibrary{calc}
\definecolor{boneC}{RGB}{255,248,220}
\newcommand{\frequentedotfill}{\leavevmode\cleaders\hbox to 0.15em{\hfil.\hfil}\hfill\kern0pt}

\newcommand{\afterreview}[1]{#1}

\newcommand{\remove}[1]{}
\newcommand{\noremove}[1]{}

\title{Separating Two Points with Obstacles in the Plane: Improved Upper and Lower Bounds}

\titlerunning{Separating Two Points with Obstacles in the Plane} %

\author{Jack Spalding-Jamieson}{Independent}{jacksj@uwaterloo.ca}{https://orcid.org/0000-0002-1209-4345}{}
\author{Anurag Murty Naredla}{University of Bonn}{anaredla@uni-bonn.de}{https://orcid.org/0000-0002-3577-903X}{}

\authorrunning{J. Spalding-Jamieson, A. M. Naredla} %

\Copyright{J. Spalding-Jamieson, A. M. Naredla} %



\category{} %

\acknowledgements{
This work originated from a problem discussion group in the Fall of 2021,
whose members collectively contributed to the discussion and development of
some of this work's preliminary results in the ``arrangement model'',
as well as a better understanding of the problem as a whole.
In addition to the authors,
this group included (in alphabetical order)
Sam Barr, Therese Biedl, Reza Bigdeli, Anna Lubiw, Jayson Lynch, Karthik Murali, and Graeme Stroud.\\
In addition, we want to thank Jeff Erickson, Elizabeth Xiao, and Da Wei Zheng for helpful discussions about homology.
}%

\EventEditors{John Q. Open and Joan R. Access}
\EventNoEds{2}
\EventLongTitle{42nd Conference on Very Important Topics (CVIT 2016)}
\EventShortTitle{CVIT 2016}
\EventAcronym{CVIT}
\EventYear{2016}
\EventDate{December 24--27, 2016}
\EventLocation{Little Whinging, United Kingdom}
\EventLogo{}
\SeriesVolume{42}
\ArticleNo{23}

\begin{document}

\maketitle

\begin{abstract}
Given two points in the plane, and a set of ``obstacles'' given as curves through the plane with assigned weights,
we consider the \defn{point-separation} problem,
which asks for a minimum-weight subset of the obstacles separating the two points.
A few computational models for this problem have been previously studied.
We give a unified approach to this problem in all models via a reduction
to a particular shortest-path problem,
and obtain improved running times in essentially all cases.
In addition, we also give fine-grained lower bounds for many cases.
\end{abstract}

\section{Introduction}
\label{sec:intro}

\begin{figure}[t]
\centering
\includegraphics[scale=0.60,page=4]{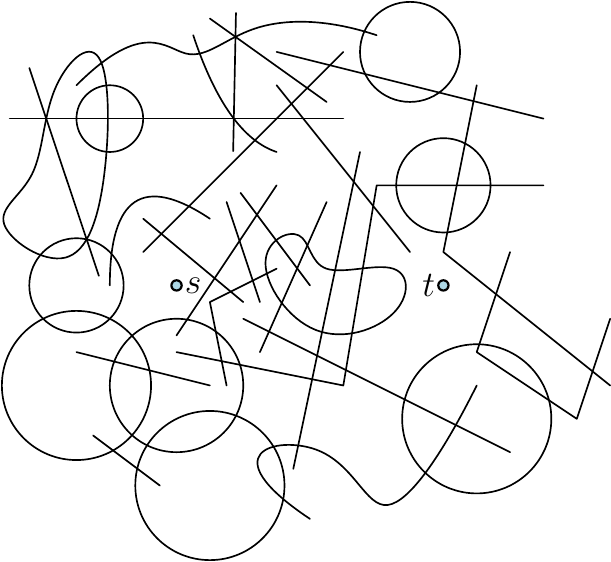}
\hspace{3em}
\includegraphics[scale=0.60,page=6]{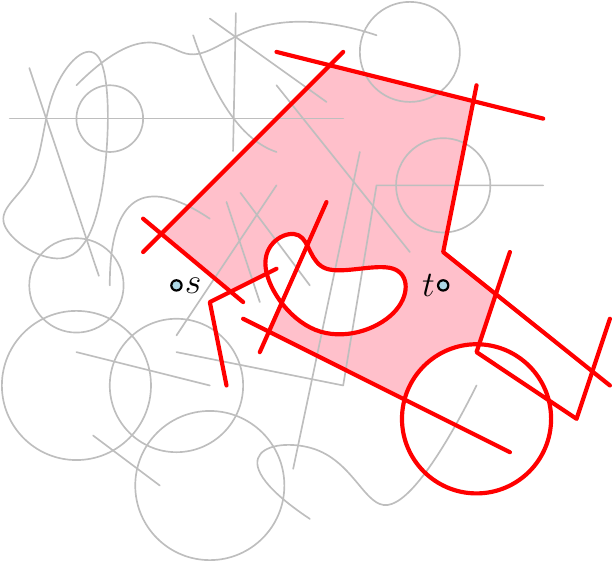}
\caption{An instance of the $(s,t)$ point-separation problem.
}
\label{fig:example-curves-1}
\end{figure}

Given points $s$ and $t$ in the plane, and a weighted set of \defn{obstacles} $\cC$ defined by simple closed curves (possibly also including their interiors),
the \defn{$(s,t)$ point-separation problem} asks
for a
minimum-weight subset $C$ of $\cC$ such that any path from $s$ to $t$ intersects some obstacle in $C$.
Equivalently, $s$ and $t$ are in different connected components of $\R^2\setminus(\cup_{\c\in C}\c)$.
We say such a subset \defn{separates} $s$ and $t$.
An example of this problem can be found in \cref{fig:example-curves-1}.

This is a natural problem that arises in various scenarios.
As a toy application of this problem: Suppose every night your dog runs from his doghouse to your backpack to eat your homework, taking an unpredictable route (making use of windows and doors). You have noticed that your dog will forget about your homework if it smells a treat.
You have a number of candidate locations to place treats,
and you'd like to ensure your dog is distracted from your homework every day using the fewest treats possible.
See \cref{fig:doggy-treats} for this example.
Similar applications also arise when considering
security (e.g., replacing the batteries in the fewest number of your dead security cameras to cover all possible paths from the entrance to your bank vault).
Additionally, $(s,t)$ point-separation has found an application as a tool for constant-factor approximation of the well-studied APX-hard problem ``barrier-resilience''~\cite{KumarLSS21}.

\begin{figure}
\centering
\begin{tikzpicture}[scale=0.28, every node/.style={font=\large,scale=0.56}]
\newcommand{\drawCommonBuilding}{%
  \draw[thick, fill=green!20] (-2,-2) rectangle (22,15);
  \node at (11,14) {Lawn};

  \fill[red!20]    (0,6)   rectangle (12,12);    %
  \fill[blue!20]   (0,0)   rectangle (12,6);      %
  \fill[yellow!20] (12,8)  rectangle (20,12);     %
  \fill[orange!20] (12,0)  rectangle (20,8);      %

  \draw[thick] (0,12) -- (2,12);
  \draw[thick,dotted] (2,12) -- (3,12);   %
  \draw[thick] (3,12) -- (6,12);
  \draw[thick,dotted] (6,12) -- (7,12);     %
  \draw[thick] (7,12) -- (9,12);
  \draw[thick] (11,12) -- (13,12);
  \draw[thick,dotted] (13,12) -- (14,12);   %
  \draw[thick] (14,12) -- (17,12);
  \draw[thick,dotted] (17,12) -- (18,12);   %
  \draw[thick] (18,12) -- (20,12);
  \node[above] at (2.5,12) {\small Window};
  \node[above] at (6.5,12) {\small Window};
  \node[above] at (13.5,12) {\small Window};
  \node[above] at (17.5,12) {\small Window};
  \draw (10,12)+(-1,0) arc (180:240:1);

  \draw[thick] (20,0) -- (20,2);
  \draw[thick,dotted] (20,2) -- (20,3);   %
  \draw[thick] (20,3) -- (20,5);
  \draw[thick,dotted] (20,5) -- (20,6);     %
  \draw[thick] (20,6) -- (20,8);
  \draw[thick,dotted] (20,8) -- (20,9);     %
  \draw[thick] (20,9) -- (20,11);
  \draw[thick,dotted] (20,11) -- (20,12);   %
  \node[rotate=90, anchor=center] at (20.5,2.5) {\small Window};
  \node[rotate=90, anchor=center] at (20.5,5.5) {\small Window};
  \node[rotate=90, anchor=center] at (20.5,8.5) {\small Window};
  \node[rotate=90, anchor=center] at (20.5,11.5){\small Window};

  \draw[thick] (0,0) -- (4,0);
  \draw[thick,dotted] (4,0) -- (5,0);    %
  \draw[thick] (5,0) -- (10,0);
  \draw[thick,dotted] (10,0) -- (11,0);  %
  \draw[thick] (11,0) -- (16,0);
  \draw[thick,dotted] (16,0) -- (17,0);  %
  \draw[thick] (17,0) -- (20,0);
  \node[below] at (4.5,0) {\small Window};
  \node[below] at (10.5,0) {\small Window};
  \node[below] at (16.5,0) {\small Window};

  \draw[thick] (0,0) -- (0,2);
  \draw[thick,dotted] (0,2) -- (0,3);   %
  \draw[thick] (0,3) -- (0,6);
  \draw[thick,dotted] (0,6) -- (0,7);   %
  \draw[thick] (0,7) -- (0,10);
  \draw[thick,dotted] (0,10) -- (0,11);   %
  \draw[thick] (0,11) -- (0,12);
  \node[rotate=90, anchor=center] at (-0.5,2.5) {\small Window};
  \node[rotate=90, anchor=center] at (-0.5,6.5) {\small Window};
  \node[rotate=90, anchor=center] at (-0.5,10.5){\small Window};

  \draw[thick] (12,0) -- (12,3.5);
  \draw[thick] (12,4.5) -- (12,9.5);
  \draw[thick] (12,10.5) -- (12,12);
  \draw (12,4.5)+(0,-1) arc (270:310:1);
  \draw (12,10.5)+(0,-1) arc (270:310:1);
  \draw[thick] (0,6) -- (5.5,6);
  \draw[thick] (6.5,6) -- (12,6);
  \draw (6.5,6)+(-1,0) arc (180:220:1);
  \draw[thick] (12,8) -- (15.5,8);
  \draw[thick] (16.5,8) -- (20,8);
  \draw (16.5,8)+(-1,0) arc (180:220:1);

  \node at (6,9)   {Kitchen};
  \node at (6,3)   {Living Room};
  \node at (16,10) {Bathroom};
  \node at (15.5,4) {Bedroom};

  \draw[thick] (2,1) rectangle (6,2.5);
  \node at (4,1.75) {Couch};

  \node[diamond, draw=red, fill=red, minimum size=2mm, inner sep=0pt] (doghouse) at (5,13.5) {};
  \node[above=0mm of doghouse, align=center] {\textbf{\small Doghouse}};
  \node[diamond, draw=red, fill=red, minimum size=2mm, inner sep=0pt] (backpack) at (13.7,6.5) {};
  \node[above=0mm of backpack, align=center] {\textbf{\small Backpack}};
}
  \begin{scope}[shift={(0,0)}]
    \drawCommonBuilding %
    \foreach \x/\y in {5/5,8/5,11/3,3/10.3,2.5/7.5,9/9.5,12/5.3,17/6,15/2,18/2,18.7/4,21/8.5,16.7/8,13/1} {%
      \draw[thick, color=magenta] (\x,\y) circle (1.5);
      \node[diamond, draw=magenta, fill=magenta!50, minimum size=2mm, inner sep=0pt] at (\x,\y) {};
    }
  \end{scope}
  \begin{scope}[shift={(25,0)}]
    \drawCommonBuilding %
    \foreach \x/\y in {11/3,12/5.3,17/6,15/2,18/2,18.7/4,16.7/8,13/1} {%
      \fill[thick, color=brown!50, fill opacity=0.85] (\x,\y) circle (1.5);
      \begin{scope}[shift={(\x,\y)}]
        \fill[color=brown] (-0.4,0.08) circle (0.1);
        \fill[color=brown] (-0.4,-0.08) circle (0.1);
        \fill[color=brown] (0.4,-0.08) circle (0.1);
        \fill[color=brown] (0.4,0.08) circle (0.1);
        \draw[color=brown,line width=0pt,fill=brown] (0.25,0.08)--(0.4-0.0707107,0.08+0.0707107)--(0.4-0.0707107,0.08);
        \draw[color=brown,line width=0pt,fill=brown] (-0.25,0.08)--(-0.4+0.0707107,0.08+0.0707107)--(-0.4+0.0707107,0.08);
        \draw[color=brown,line width=0pt,fill=brown] (0.25,-0.08)--(0.4-0.0707107,-0.08-0.0707107)--(0.4-0.0707107,-0.08);
        \draw[color=brown,line width=0pt,fill=brown] (-0.25,-0.08)--(-0.4+0.0707107,-0.08-0.0707107)--(-0.4+0.0707107,-0.08);
        \fill[color=brown] (-0.4,0.08) rectangle (0.4,-0.08);
      \end{scope}
    }
  \end{scope}
  \node[draw, align=left, anchor=west] at (-2,16.5) {%
    \parbox{8.5cm}{%
    \begin{tabular}{l|cccc}
      \textbf{Item} & Eating locations (\tikz[baseline=-0.5ex]\node[diamond, draw=magenta, fill=magenta!50, minimum size=2mm, inner sep=0pt]{};) & Walls & Couch & Windows \\
      \hline
      \textbf{Cost} & 1 & 0 & 0 & $\infty$ \\
    \end{tabular}
    }
  };
\end{tikzpicture}
\caption{An example of the point-separation problem applied to placing distracting doggy treats around a house, with the minimum-weight solution on the right.}
\label{fig:doggy-treats}
\end{figure}
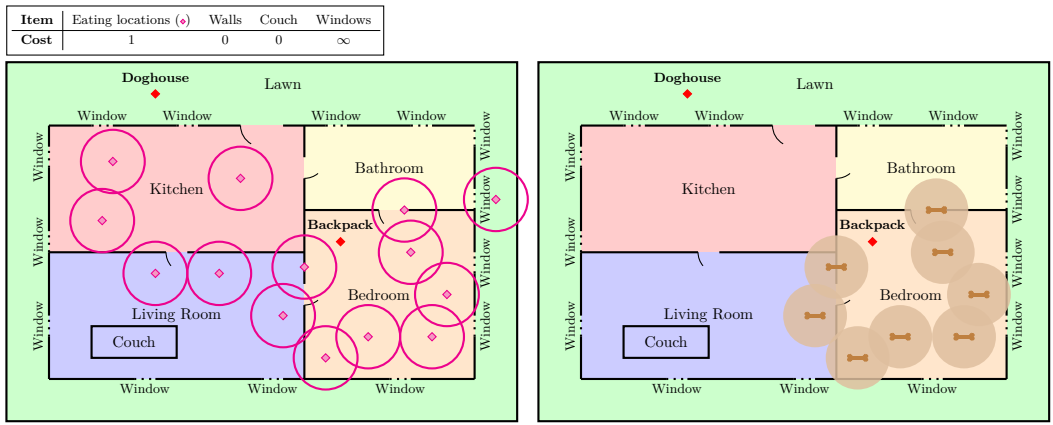

For brevity, we will henceforth say ``curve''
to mean a simple closed curve in $\R^2\setminus\{s,t\}$.
The algorithmic complexity of the $(s,t)$ point-separation problem depends on the chosen computational model of the obstacles.
Previous works use a few different models.
We categorize and name the different classes of models used as follows:
\begin{itemize}
    \item\defn{Specific Obstacle-Type Models}: Assume the set of curves $\cC$ takes on a special form with a standard representation, such as a set of disks, circles, or line segments.
    \item\defn{The Oracle-based Intersection Graph Model}: Assume the existence of several $O(1)$-time oracles that would allow the computation of the \defn{intersection graph} $G$ of $\cC$ (the graph whose vertices are exactly the curves $\cC$ and whose edges correspond to pairwise intersections of these curves),
    as well as some additional information for each edge related to $s$ and $t$ (detailed in \cref{sec:homology-informal}).
    \item\defn{The Arrangement Model}: Assume the arrangement of $\cC$ is provided as input, in the form of a plane (multi-)graph, with the faces corresponding to $s$ and $t$ labelled. This is a more graph-theoretic formulation.
    See \cref{fig:example-curves-arrangement} for an example of such an arrangement.
\end{itemize}

\begin{figure}
\centering
\includegraphics[scale=0.60,page=3,valign=M]{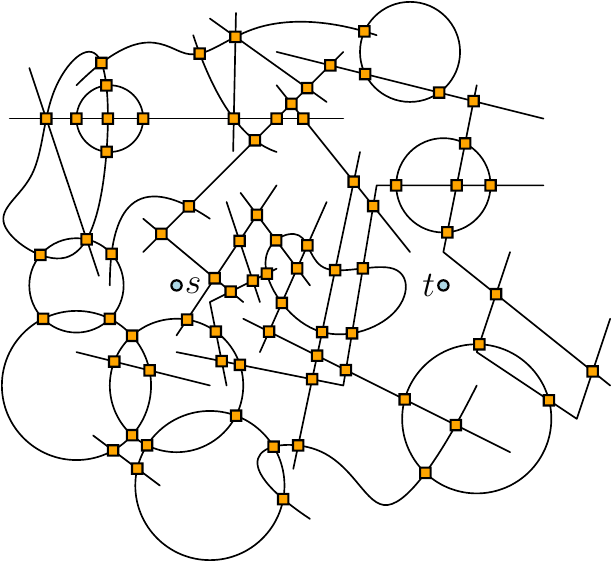}
\hspace{1.5em}
\includegraphics[scale=0.60,page=3,valign=M]{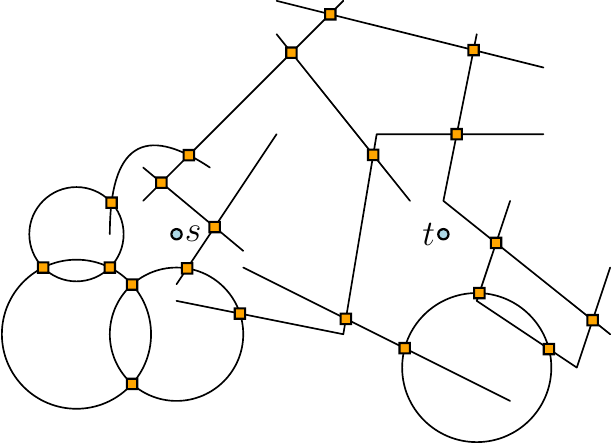}
\caption{An arrangement representation of the instance in \cref{fig:example-curves-1},
as well as an arrangement of a smaller instance.}
\label{fig:example-curves-arrangement}
\end{figure}

Prior works discussing this problem each focused on only one of these paradigms,
and the names for each paradigm are our own.
We will not assume general position,
so
no paradigm is strictly more general than the others,
since there are arrangements of $n$ obstacles
with $\Theta(n^2)$ pairwise intersections,
but only $O(n)$ unique intersection points.

In our work, our key method is to frame all models of this problem in terms
of a ``homology cover''.
Homology is a very broad topic in algebraic topology, which we will not attempt to summarize
in this paper. We aim our paper at a typical computational geometry audience, and we will not assume prior knowledge of homology.
We will present the necessary aspects in
\cref{sec:homology-informal}.

\subsection{Prior Work (Brief)}

In this subsection, we very briefly outline some key aspects of prior work.
A significantly extended form of this section is in
\arxivappendix{\cref{sec:prior-work}}.

Most importantly for our methods,
Kumar, Lokshtanov, Saurabh and Suri~\cite[Section 6]{KumarLSS21}
describe an algorithm that runs in polynomial time
in the \emph{arrangement model}.
Their algorithm implicitly makes use of the homology cover
(perhaps unintentionally),
in the same sense that we will use it.
In fact, their algorithm \afterreview{is in some ways} analogous to
an algorithm of Chambers, Erickson, Fox, and Nayyeri~\cite{ChambersEFN23}
for minimum-cuts (and maximum-flow) in surface graphs.
These two algorithms are the main inspiration for our approach to the $(s,t)$ point-separation at a high-level,
in \emph{all} models.

\subsection{Results and Organization}
\label{subsec:results}

We obtain improved algorithms for the $(s,t)$-point separation problem
in several cases, which we outline in
\cref{tab:positive-result-summary}.
As mentioned earlier, all of our positive results
make use of a reduction to a shortest-path problem
in the so-called ``homology cover''.
We discuss this formulation in \cref{sec:homology-informal},
and then again more rigorously in \arxivappendix{\cref{sec:homology-formal}}.
Using this reduction, the upper bounds are then given in \cref{sec:algorithms}.

\begin{table}[ht]
\centering
\begin{tabular}{lcrrr}
\toprule
\textbf{Model} & \textbf{Weights} &
\multicolumn{1}{c}{\textbf{Old}} &
\multicolumn{1}{c}{\textbf{New}} &
\multicolumn{1}{c}{\textbf{Cond.~L.B.}} \\
\midrule
Oracle & Yes                  & $O(n^3)$~\cite{CabelloG16} & $\widetilde O(n^{(3+\omega)/2})$ (Thm.~\ref{thm:better-than-n3}) & $\Omega^*\left(n^2\right)$\arxivonly{
 (Thm.~\ref{thm:lb-segments-clique})}\\[3pt]
Oracle & No                & $O(n^3)$~\cite{CabelloG16} & $O(n^\omega\log n)$ (Thm.~\ref{thm:general-unweighted}) & $\Omega^*(n^{3/2})$\arxivonly{ (Cor.~\ref{thm:unweighted-lb})}\\[3pt]
Arrangement & Yes             & $O(km^2\lg k)$~\cite{KumarLSS21} & $O(km+k^2\lg k)$ (Thm.~\ref{thm:not-match-cg16-arrangement}) & $\Omega^*\left(km\right)^\dag$\arxivonly{ (Cor.~\ref{cor:lb-polylines})}\\[3pt]
Segments & No      & $O(n^3)$~\cite{CabelloG16} & $O(n^{7/3}\log^{1/3}n)$ (Cor.~\ref{cor:restricted-si-actual}) & $\Omega^*(n^{3/2})$\arxivonly{ (Thm.~\ref{thm:unweighted-lb})}\\[3pt]
{\begin{socgtabular}{@{}l@{}}Axis-aligned\\segments\end{socgtabular}} & No & $O(n^3)$~\cite{CabelloG16} & $O(n^2\log\log n)$ (Cor.~\ref{cor:restricted-si-actual}) & None\\[9pt]
Unit Disks & No                 & $O(n^2\log^3 n)$~\cite{CabelloM18} & $O(n^2\log n)$ (Cor.~\ref{cor:restricted-si-actual}) & None\\[3pt]
Disks & No                 & $O(n^3)$~\cite{CabelloG16} & $O(n^2\log n)$ (Cor.~\ref{cor:restricted-si-actual}) & None\\[3pt]
\begin{socgtabular}{@{}l@{}}$O(1)$-length\\polylines\end{socgtabular}& No & $O(n^3)$~\cite{CabelloG16} &
$O(n^{7/3}\log^{1/3}n)$ (Cor.~\ref{cor:restricted-si-actual})
& $\Omega^*\left(n^{3/2}\right)$\arxivonly{ (Thm.~\ref{thm:unweighted-lb})}\\[9pt]
\begin{socgtabular}{@{}l@{}}$O(1)$-length\\rectilinear\\polylines\end{socgtabular}& No & $O(n^3)$~\cite{CabelloG16} &
$O(n^2\log\log n)$ (Cor.~\ref{cor:restricted-si-actual})
& $\Omega^*(n^{3/2})$\arxivonly{ (Thm.~\ref{thm:unweighted-lb})}\\
Segments & Yes      & $O(n^3)$~\cite{CabelloG16} & $\widetilde O(n^{7/3})$ (Cor.~\ref{cor:separation-biclique-cover-times}) & $\Omega^*\left(n^2\right)$\arxivonly{ (Thm.~\ref{thm:lb-segments-clique})}\\[3pt]
\begin{socgtabular}{@{}l@{}}Axis-aligned\\segments\end{socgtabular} & Yes & $O(n^3)$~\cite{CabelloG16} & $\widetilde O(n^2)$ (Cor.~\ref{cor:separation-biclique-cover-times}) & None\\[9pt]
\begin{socgtabular}{@{}l@{}}$O(1)$-length\\polylines\end{socgtabular}& Yes & $O(n^3)$~\cite{CabelloG16} & $\widetilde O(n^{7/3})$ (Cor.~\ref{cor:separation-biclique-cover-times}) & $\Omega^*\left(n^2\right)$\arxivonly{ (Thm.~\ref{thm:lb-segments-clique})}\\[9pt]
\begin{socgtabular}{@{}l@{}}$O(1)$-length\\rectilinear\\polylines\end{socgtabular}& Yes & $O(n^3)$~\cite{CabelloG16} & $\widetilde O(n^2)$ (Cor.~\ref{cor:separation-biclique-cover-times}) & $\Omega^*(n^2)$\arxivonly{ (Thm.~\ref{thm:lb-segments-clique})}\\
\bottomrule
\end{tabular}
\caption{The time complexities of our algorithms for various obstacle models. In all cases, $n:=|\cC|$. For the arrangement model, $k$ denotes the vertex count of the arrangement, and $m$ is the number of obstacle-vertex incidences (so $m\geq k$).
The notation $\Omega^*(\cdot)$ hides sub-polynomial factors.
The notation $^\dag$ denotes that this is only true for one particular mutual dependence of $k$, $m$, and $n$ (and is a slight abuse of notation).
$\omega$ is the matrix-multiplication exponent ($\omega<2.371339$~\cite{alman2024asymmetryyieldsfastermatrix}).
}
\label{tab:positive-result-summary}
\end{table}

We also obtain several fine-grained lower bounds in \cref{sec:lower-bounds}.
Our lower bounds also all have a shared foundation,
which will be stated
in \cref{thm:k-cycle-reduction}.
The resulting lower bounds are based on a few different hypotheses,
and we give their details in
\arxivappendix{\cref{sec:lb-details}}.

\section{Homology and Obstacles [Informal]}
\label{sec:homology-informal}

In this section, we will explain the main topological tools we will use to approach the $(s,t)$ point-separation problem.
However, this will be an information section, not requiring prior knowledge of any aspect of topology.
Rather, we will give an equivalent formulation of the important pieces using simpler tools from computational geometry.
We give a more formal treatment in \arxivappendix{\cref{sec:homology-formal}}.
Throughout this section, we will make (implicit) assumptions of general position and such,
but the deferred formal treatment does not need these.

At a high-level, there are two main steps to the constructions of our approach.
First, we will discuss what it means for a \emph{simple curve} to separate $s$ and $t$,
and what tools exist to classify such curves.
Second, we will discuss what it means for a \emph{set of obstacles} to separate $s$ and $t$.
\afterreview{That is, when does the union of the obstacles \emph{contain} a simple curve separating $s$ and $t$?}

The simplest possible case of asking whether a simple curve separates $s$ and $t$ is
characterized by the
\defn{point-in-polygon} problem, which asks:
Given a point $p$ and a (simple) polygon $P$, is $p$ inside $P$?
In fact, this is equivalent to asking if $P$ separates $p$ and the point at infinity.
There is a folklore algorithm for this problem that
takes any ray $r$ starting at $p$,
and counts the number of intersections between $r$ and $P$.
Then, $P$ contains $p$ if and only if the number of intersections is odd.
In fact,
essentially the same algorithm can be used to test whether or not $P$ separates
two points $p_1$ and $p_2$.
Rather than using a ray, take the segment $\overline{p_1p_2}$.
Count the number of intersections between $\overline{p_1p_2}$ and $P$.
The count is odd if and only if $P$ separates $p_1$ and $p_2$.
Moreover, there is in fact nothing special about rays or line segments in this problem.
\emph{Any} (closed) path between $p_1$ and $p_2$ would cross
$P$ the same number of times, modulo $2$.
See \cref{fig:point-in-polygon} for examples for these algorithms.

\begin{figure}
    \centering
    \includegraphics[scale=0.9,page=1,valign=t]{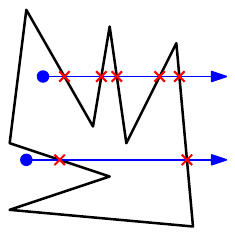}
    \hspace{3em}
    \includegraphics[scale=0.9,page=2,valign=t]{point-in-polygon-figure}
    \hspace{3em}
    \includegraphics[scale=0.9,page=3,valign=t]{point-in-polygon-figure}
    \caption{A demonstration of the algorithm for the point-in-polygon problem (left), as well as the problem of testing whether or not a polygon separates two points (middle).
    On the right, alternative paths for the separation problem are given.}
    \label{fig:point-in-polygon}
\end{figure}

All three of these algorithms are equivalent in a sense.
In fact, there is a further generalization:
Given two points $s$ and $t$ on the sphere, a simple and closed curve $C$ (not covering $s$ or $t$),
and any $s-t$ path $\pi$,
$C$ separates $s$ and $t$ if and only if $\pi$ crosses $C$ an odd number of times
(regardless of the choice of $\pi$).
Since the extended plane is homeomorphic to the sphere,
a ray in the plane corresponds to a (simple) path in the sphere.
See \cref{fig:point-pair-curve-separation} for an example.

\begin{figure}
    \centering
    \includegraphics[scale=0.9,page=4]{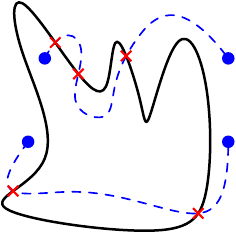}
    \caption{Examples of point-pairs separated or not separated by a given closed path between the pair, and the corresponding intersections of those curves.}
    \label{fig:point-pair-curve-separation}
\end{figure}

The problems we have discussed so far are all completely static,
and the methods do not provide much structure for solving more difficult problems.
One of the most common ways to extend the point-in-polygon problem
is to fix the polygon (or curve) $P$,
and aim to support fast queries of points.
This problem is known as ``point-location'',
and it is well-studied in computational geometry~\cite{BergCKO08}.
However, we want a different sort of structure:
We have a fixed pair of points $s$ and $t$,
and we wish to classify the curves that separate them.
Since we have fixed $s$ and $t$, we can also fix the path $\pi$
between them -- in most cases, we will use the line segment $\overline{st}$.
Then, the problem of classifying curves that separate $s$ and $t$
becomes the problem of classifying curves that cross $\overline{st}$ an odd number of times.

It will be helpful for demonstration to make some transformations to the space we work in.
That is, we will perform a sequence of homeomorphisms.
We start with the extended plane with the two marked points $s$ and $t$,
and the line segment $\overline{st}$
(see \cref{fig:annulus-construction-a}).
No curves we will be classifying cross $s$ or $t$,
so we may assume there are punctures
at $s$ and $t$.
We can enlarge these punctures into holes with a homeomorphism.
Next, since we have the extended plane,
we can make one of the punctures the outer face,
obtaining an annulus
(see \cref{fig:annulus-construction-b} and \cref{fig:annulus-construction-c}).
In performing these steps, the line segment $\overline{st}$ becomes a path between the inner and outer boundaries of the annulus.

\begin{figure}
    \centering
    \begin{subfigure}[t]{0.28\textwidth}
        \centering
        \includegraphics[scale=0.6,page=7]{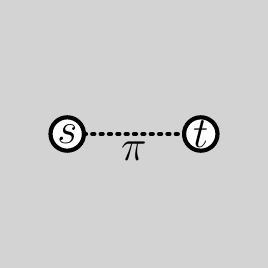}
        \caption{$s$ and $t$ are punctures, $\pi$ is a path between them.}
        \label{fig:annulus-construction-a}
    \end{subfigure}
    \hfill
    \begin{subfigure}[t]{0.28\textwidth}
        \centering
        \includegraphics[scale=0.6,page=8]{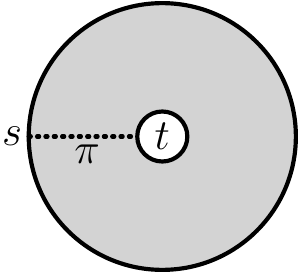}
        \caption{$s$ and $t$ are enlarged to become holes, and the space is an annulus.}
        \label{fig:annulus-construction-b}
    \end{subfigure}
    \hfill
    \begin{subfigure}[t]{0.28\textwidth}
        \centering
        \includegraphics[scale=0.6,page=9]{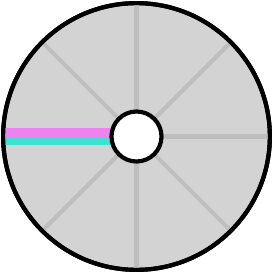}
        \caption{
        The two sides of $\pi$ can be coloured to aid later notation.
        }
        \label{fig:annulus-construction-c}
    \end{subfigure}
    \caption{A demonstration of how the (extended) plane with two punctures is homeomorphic to the annulus.}
    \label{fig:annulus-construction}
\end{figure}

We now present a method for constructing an important
space called
the \defn{homology cover}\footnote{
We use the term ``homology cover'' to refer to the one-dimensional $\mathbb{Z}_2$-homology cover
of the annulus.
}.
Take the specified path $\pi$ between the two boundaries
(see \cref{fig:homology-cover-construction-a})
and slice it open
(see \cref{fig:homology-cover-construction-b}).
Then, create a second copy of the sliced annulus,
and glue them together
in a way that matches up the orientations of the sliced ends
(see \cref{fig:homology-cover-construction-c}).

\begin{figure}
    \centering
    \begin{subfigure}[t]{0.28\textwidth}
        \centering
        \includegraphics[scale=0.6,page=13]{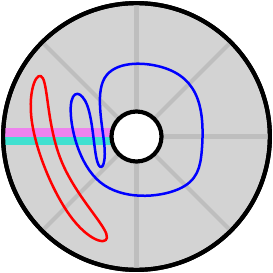}
        \caption{An annulus with a specified path $\pi$.}
        \label{fig:homology-cover-construction-a}
    \end{subfigure}
    \hfill
    \begin{subfigure}[t]{0.28\textwidth}
        \centering
        \includegraphics[scale=0.6,page=14]{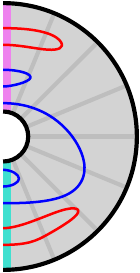}
        \caption{The annulus is sliced open along the specified path $\pi$.}
        \label{fig:homology-cover-construction-b}
    \end{subfigure}
    \hfill
    \begin{subfigure}[t]{0.28\textwidth}
        \centering
        \includegraphics[scale=0.6,page=15]{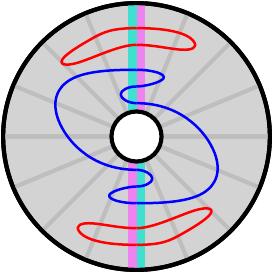}
        \caption{Two copies of the sliced annulus are glued together.}
        \label{fig:homology-cover-construction-c}
    \end{subfigure}
    \caption{The ``cut and glue'' construction of the homology cover,
    as well as how it maps a closed curve that separates the two boundaries (blue)
    and one that does not (red).}
    \label{fig:homology-cover-construction}
\end{figure}

Alternatively, an equivalent construction is to
create two copies of the original space
(whether that be the extended plane or the annulus),
and use each side of the path from $s$ to $t$
as a (separate) ``portal'' between the two copies.
A more general form of this ``portal'' idea has been studied
in the form of ``portalgons''~\cite{loffler_et_al,ophelders2024shortest},
of which the homology cover is essentially a special case.
See \cref{fig:homology-cover-construction-portal} for an example of this construction.

\begin{figure}
    \centering
    \includegraphics[scale=0.6,page=16]{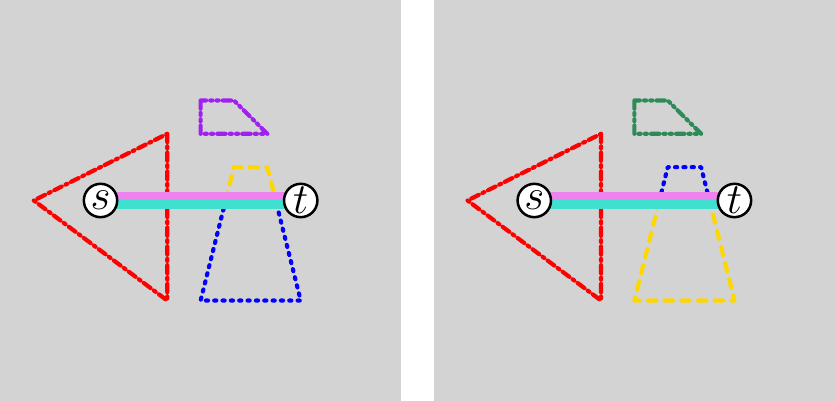}
    \caption{The ``portal'' construction of the homology cover. Each colour (or dot/dash pattern) is a single closed curve in the homology cover. The two solid lines are the portals.}
    \label{fig:homology-cover-construction-portal}
\end{figure}

One important aspect of this construction is that
it involves the connection of two identical copies
of the annulus (i.e., the extended plane with punctures $s$ and $t$).
In \cref{fig:homology-cover-construction},
we also show how this construction transforms
two curves -- one separating the two boundaries, and one not separating them.
The structure we will make use of
to study curves separating $s$ and $t$ in the plane ultimately stems from an important set of facts:
\begin{fact}
\label{fact:homology-cover-utility-informal}
For a simple curve $C$ in the annulus (or the plane)
that gets mapped to the set $C'$ in the homology cover,
and a point $p$ along $C$ that gets mapped to corresponding points $p_1$ and $p_2$ in the homology cover,
the following are all equivalent:
\begin{itemize}
    \item $C$ separates $s$ and $t$.
    \item $C'$ separates the two boundaries in the homology cover.
    \item $C'$ has one connected component (i.e., it is one closed curve instead of two).
    \item The points $p_1$ and $p_2$ are connected by a path through $C'$.
\end{itemize}
\end{fact}

It is this last characterization that will be the most critical for obstacles.
In particular, a consequence of this characterization is that
if two corresponding points $p_1$ and $p_2$ in the homology cover (corresponding in the sense that they are identical points in different copies of the annulus/plane)
have a simple path $\pi$ between them,
then that path can be mapped to a closed curve separating $s$ and $t$.

\subsection{Obstacles and the Homology Cover}

We now have tools for characterizing \emph{closed curves} that separate $s$ and $t$.
We'd now like to characterize \emph{sets of obstacles} that separate $s$ and $t$.
In other words, we'd like to characterize unions of closed curves (obstacles)
that contain a closed curve separating $s$ and $t$.
We'll give two different tools for this,
first for the arrangement model, then for the oracle model.
For simplicity, we will work exclusively with obstacles that do not \emph{individually} separate
$s$ and $t$ (that is, obstacles that get mapped to \emph{two} separate closed curves in the homology cover).
We will reduce to this case algorithmically at the start of
\cref{sec:algorithms}.

In the arrangement model of the $(s,t)$ point-separation problem,
we are given a plane graph $D$ (the arrangement)
with specified faces $s$ and $t$ (can be obtained from points via point-location),
and a set of obstacles $\sigma$, each given as a connected subgraph.
The homology cover has a simple graph-theoretic interpretation in this framework:
Take any (simple) dual-path $\pi$ of faces from $s$ to $t$ -- this is will serve as the ``portal''.
Next, create two copies $D_1,D_2$ of $D$, each with the faces $s$ and $t$ removed.
We will create a combined graph $\overline D$ starting from the union of $D_1$ and $D_2$.
For each edge $e=uv$ used in the dual-path $\pi$,
delete $e$ from both $D_1$ and $D_2$ (inside $\overline D$).
Denote the corresponding vertices of $u$ and $v$ in each of $D_1$ and $D_2$ as $u_1,v_1$ and $u_2,v_2$, respectively.
Then create new edges $u_1v_2$ and $u_2v_1$ in $\overline D$.
This form of the homology cover is visualized in \cref{fig:example-curves-arrangement-cover}.

\begin{figure}[ht]
    \centering
    \begin{subfigure}[t]{0.30\textwidth}
        \centering
        \includegraphics[scale=0.40,page=1,valign=M]{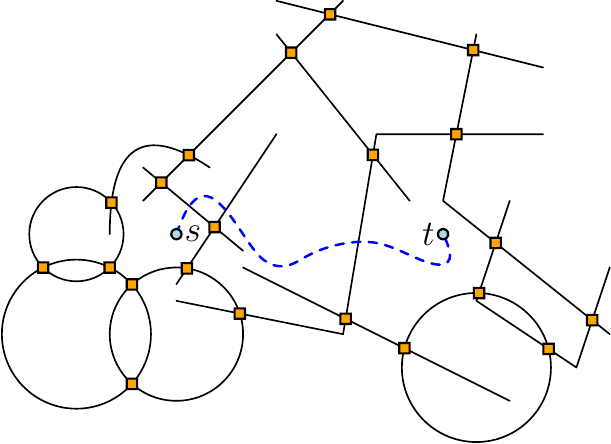}
        \caption{The arrangement with a specified path $\pi$.}
    \end{subfigure}
    \hfill
    \begin{subfigure}[t]{0.30\textwidth}
        \centering
        \includegraphics[scale=0.40,page=5,valign=M]{example-curves-small}
        \caption{A simple separating cycle in the arrangement.}
        \label{fig:example-curves-arrangement-cover:sep}
    \end{subfigure}
    \hfill
    \begin{subfigure}[t]{0.30\textwidth}
        \centering
        \includegraphics[scale=0.40,page=9,valign=M]{example-curves-small}
        \caption{The edges of the arrangement crossing $\pi$ are deleted.}
    \end{subfigure}\\
    \begin{subfigure}[t]{0.47\textwidth}
        \centering
        \includegraphics[scale=0.50,page=11,valign=M]{example-curves-small}
        \caption{Two copies of are created, and the deleted edges are added as ``crossing'' edges (in purple).}
    \end{subfigure}
    \hfill
    \begin{subfigure}[t]{0.47\textwidth}
        \centering
        \includegraphics[scale=0.50,page=13,valign=M]{example-curves-small}
        \caption{A path between two copies of the same original vertex that would project onto the simple separating cycle in the original space from
        \cref{fig:example-curves-arrangement-cover:sep}.}
    \end{subfigure}
    \caption{The construction of the homology cover for the second arrangement given in \cref{fig:example-curves-arrangement}.}
    \label{fig:example-curves-arrangement-cover}
\end{figure}

With this particular arrangement structure
Kumar, Lokshtanov, Saurabh and Suri~\cite[Section 6]{KumarLSS21}
built a graph
using one vertex
per obstacle-vertex incidence
(in each copy of the plane graph).
We will
build a smaller graph of similar form to theirs.
First, for each obstacle $C\in\sigma$ (which induces a subgraph of $D$),
pick some arbitrary \defn{canonical point} $x\in C$.
Since this choice is arbitrary, it will sometimes
be useful to assume it is a specific point -- this will primarily be useful for specific obstacle types.
Assume that $x$ is also a vertex in the subgraph of $D$ induced by $C$
(and if it is not, modify $D$ so that it is, with either a \afterreview{subdivision} or a new degree-$1$ vertex).
Note that since $C$ is connected in $D$, every other vertex $x'\in C$
is connected to $x$ by only edges in $C$.
Given a plane graph $D$ with faces $s$ and $t$\afterreview{,} a dual-path $\pi$,
obstacles $\sigma$, and canonical points for each obstacle, the \defn{auxiliary graph} $H$
is a bipartite graph constructed as follows:
\begin{itemize}
    \item The first set of vertices are the copies of arrangement vertices in the homology cover.
    For a vertex $v\in V(D)$, we denote these two copies $v^+$ and $v^-$.
    \item The second set of vertices are the copies of the obstacles/canonical points in the homology cover.
    For an obstacle $\gamma\in\sigma$, we denote these two copies as $(\gamma,-)$ and $(\gamma,+)$.
    \item The values $-$ and $+$ in each of the above are \defn{indicator bits}, and can be stored as $0$ and $1$, respectively.
    \item A vertex copy $v^{b_1}$ has an edge connecting it to a canonical point copy $(\gamma,b_2)$ when:
    \begin{itemize}
        \item $v\in\gamma$ (that is, the point $v$ \afterreview{belongs to} the obstacle $\gamma$ in the arrangement), and\dots
        \item The canonical point $x$ of $\gamma$ has a path through the edges of $\gamma$ to $v$
        whose intersection with the edges of $\pi$ has size $b_1+b_2$ modulo $2$
        (i.e., if $b_1=b_2$, the size must be even, and if $b_1\neq b_2$, the size must be odd).
        Equivalently, the projected connected component of $\gamma$ containing the canonical point copy $(\gamma,b_2)$
        also contains $v^{b_1}$.
    \end{itemize}
\end{itemize}

We visualize an example of the auxiliary graph in \cref{fig:example-curves-auxiliary-graph}.

\begin{figure}[t]
    \centering
    \begin{subfigure}[t]{0.47\textwidth}
        \centering
        \includegraphics[scale=0.50,page=19,valign=M]{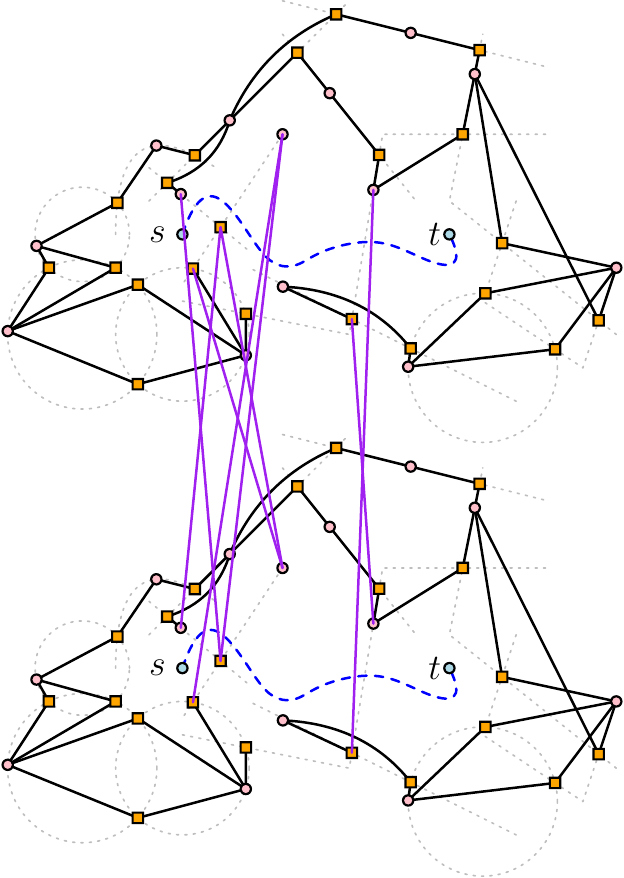}
        \caption{A canonical point is added along each obstacle, and edges are created from the canonical point to each arrangement vertex along the obstacle, subject to the parity conditions.}
    \end{subfigure}
    \hfill
    \begin{subfigure}[t]{0.47\textwidth}
        \centering
        \includegraphics[scale=0.50,page=20,valign=M]{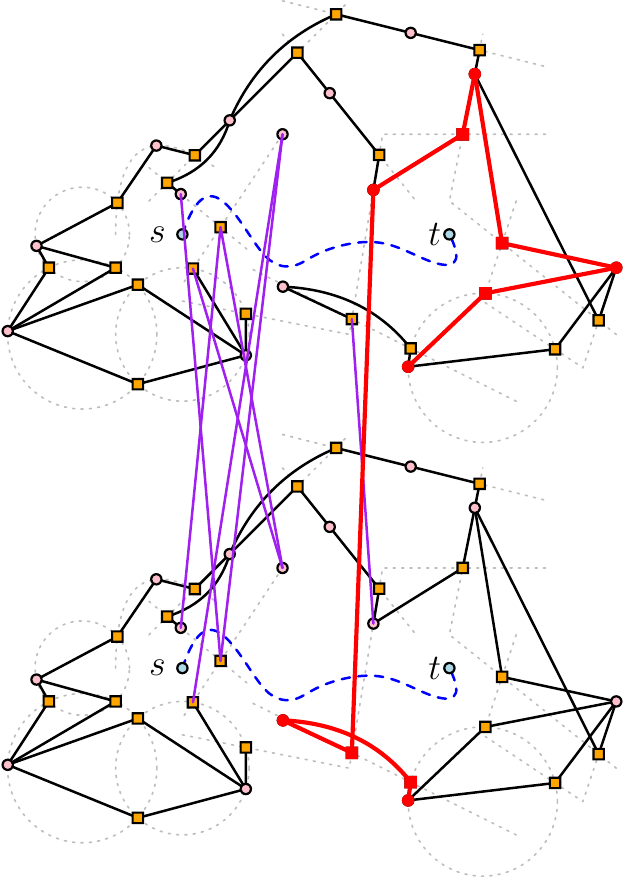}
        \caption{A shortest-path between the two copies of an obstacle's canonical point. The obstacles whose canonical points lie in this path form a set of obstacles separating $s$ and $t$.}
        \label{fig:example-curves-auxiliary-graph:sp}
    \end{subfigure}
    \begin{subfigure}[t]{0.9\textwidth}
        \centering
        \includegraphics[scale=0.50,page=6,valign=M]{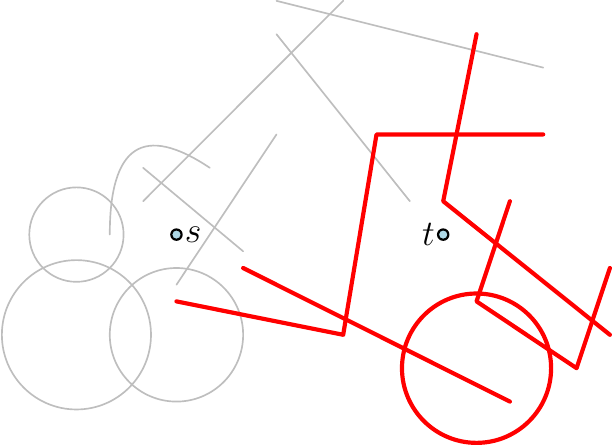}
        \caption{The set of obstacles induced by the shortest path in \cref{fig:example-curves-auxiliary-graph:sp}.}
    \end{subfigure}
    \caption{The construction of the auxiliary graph for the homology cover given in 
    \cref{fig:example-curves-arrangement-cover}.}
    \label{fig:example-curves-auxiliary-graph}
\end{figure}

The auxiliary graph is useful because the \afterreview{shortest path} between any two corresponding vertex copies
or (separately) any two corresponding canonical point copies both correspond exactly to the solution to the $(s,t)$ point-separation problem.
The high-level construction to prove this fact is as follows:
\begin{itemize}
    \item Suppose there is a set of obstacles $C$. We'd like to determine when $C$ separates $s$ and $t$.
    \item Denote the set of copies of the obstacles in $C$ in the homology cover as $C'$,
    so that every element of $C'$ is a closed curve representing one of the two connected components
    of an element of $C$ projected into the homology cover.
    \item By \cref{fact:homology-cover-utility-informal}, we deduce that $C$ separates $s$ and $t$
    if and only if there is a path from $v^+$ to $v^-$ along the union of obstacle copies in $C'$,
    for some arrangement vertex $v$ among pairs of elements in $C$ only.
    \item Moreover, since each obstacle is connected, we may further assume that any such path visits each obstacle copy at most once,
    and moreover that it only visits each \emph{obstacle} at most once.
    \item We may further assume that any such path visits the canonical point copy of each such obstacle copy,
    by inserting a path to and then from the canonical point copy while the obstacle copy is visited.
    Note that the resulting path may not be simple in the plane.
    \item Such paths also directly correspond to paths in the auxiliary graph, using only ``obstacle vertices'' corresponding to elements of $C$ and the ``arrangement vertices'' incident to pairs of elements in $C$.
    \item With appropriate weights, the problem of finding the minimum-weight set $C$ with this property
    then reduces to the problem of finding the pair $v^+,v^-$ with the shortest distance in the auxiliary graph,
    so this is a shortest-path problem!
    \item A similar argument can show that an alternative algorithmic formulation is to find the pair $(\gamma,-),(\gamma,+)$ with the shortest-distance in the auxiliary graph.
\end{itemize}
This essentially completes the set of tools necessary for the arrangement model.
These arguments are given in more detail in \arxivappendix{\cref{sec:homology-formal}}.

\begin{figure}[ht]
    \centering
    \begin{subfigure}[t]{0.47\textwidth}
        \centering
        \includegraphics[scale=0.5,page=14,valign=c]{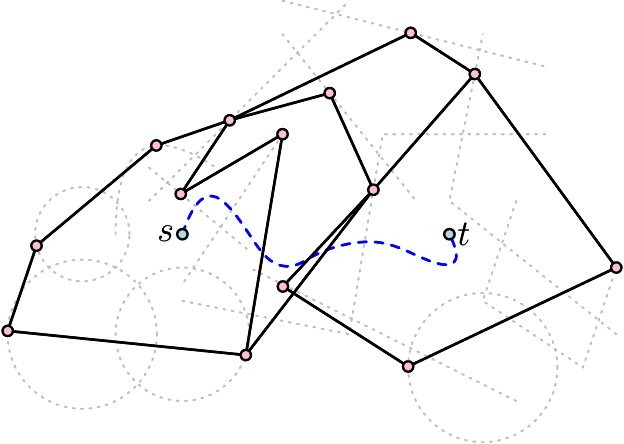}
        \caption{The intersection graph with the path $\pi$.}
    \end{subfigure}
    \hfill
    \begin{subfigure}[t]{0.47\textwidth}
        \centering
        \includegraphics[scale=0.5,page=15,valign=c]{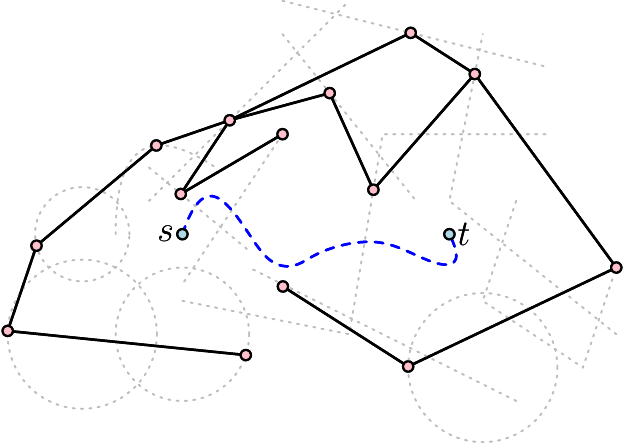}
        \caption{The intersection graph after deleting the edges
            whose underlying paths cross $\pi$ an odd number of times.}
    \end{subfigure}\\
    \begin{subfigure}[t]{0.47\textwidth}
        \centering
        \includegraphics[scale=0.5,page=17,valign=c]{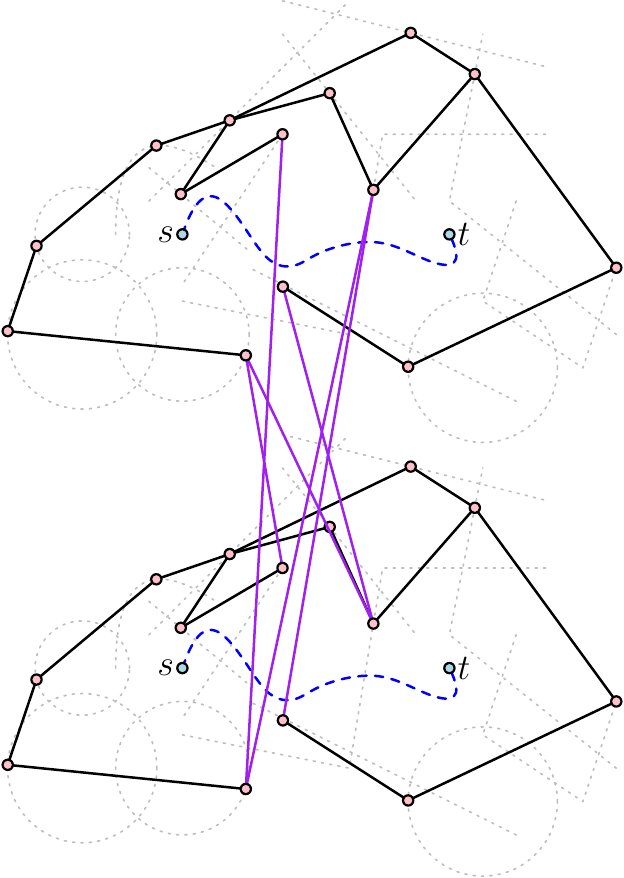}
        \caption{Two copies of the intersection graph are connected with ``crossing'' edges (in purple).}
    \end{subfigure}
    \hfill
    \begin{subfigure}[t]{0.47\textwidth}
        \centering
        \includegraphics[scale=0.5,page=18,valign=c]{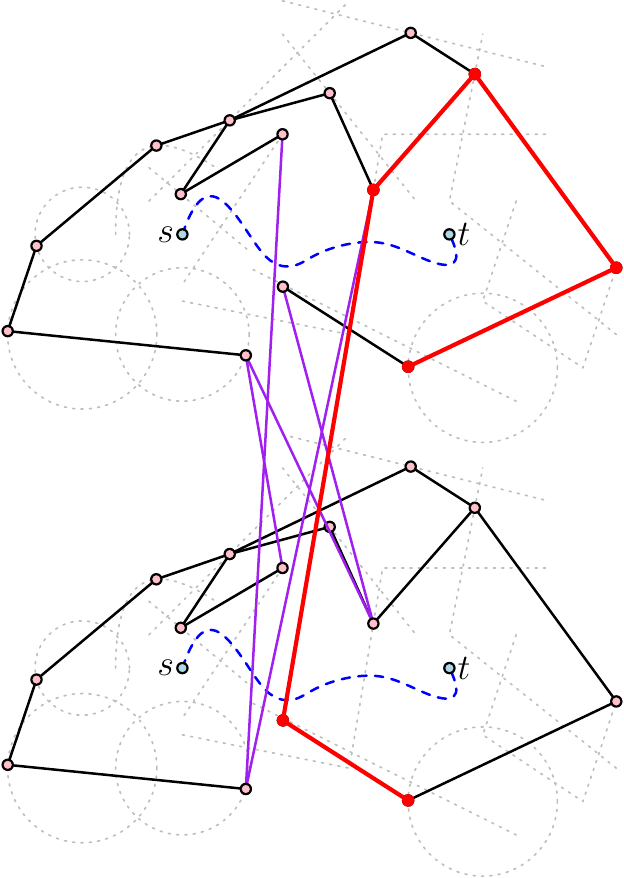}
        \caption{A shortest-path in the homology cover between two copies of an obstacle.}
    \end{subfigure}
    \caption{How the intersection graph can be transformed into the intersection graph in the homology cover.}
    \label{fig:example-curves-isect-cover}
\end{figure}

For the oracle model, the purpose of the oracles will be to construct
the geometric intersection graph of the obstacle \emph{copies} in the homology cover
(which we will denote as $\overline G$).
We will call $\overline G$ the \defn{intersection graph in the homology cover}.
See \cref{fig:example-curves-isect-cover} for an example.
We will present two different constructions of this graph.
They are equivalent, and each of them is quite simple,
but they will serve two different purposes:
The first will clarify which types of oracle queries are necessary,
and provide an algorithmic construction.
The second will instead build on the tools we have for the arrangement model,
proving the correctness of another shortest-paths approach.

Let $\cC$ denote the set of obstacles in the plane.
We assume the canonical points for each obstacle are given (or implied) -- they will be used by the oracle.
We also assume some simple $s-t$ path $\pi$ is fixed -- this will also be used by the oracle.
In most cases, $\pi$ will be $\overline{st}$.
The first construction of the graph is as follows:
\begin{itemize}
    \item For each obstacle $\gamma\in\cC$, create two vertices $(\gamma,-)$ and $(\gamma,+)$.
    \item For each pair of obstacles $\gamma_1,\gamma_2$
    with canonical points $x_1,x_2$,
    and values $b_1,b_2\in\{-,+\}$,
    we connect vertices $(\gamma_1,b_1)$ and $(\gamma_2,b_2)$ by an edge if
    there is a path from $x_1$ to $x_2$ in $\gamma_1\cup\gamma_2$
    crossing $\pi$ exactly $k$ times, for some $k$ where $k\equiv b_1+b_2\pmod 2$.
\end{itemize}
Testing for this condition is actually the \emph{only} type of oracle query needed
(although we will also require support for the case when $\gamma_1=\gamma_2$ later).
Cabello and Giannopoulos~\cite{CabelloG16} used a larger set of query types,
but together they can be used to perform this type
by carefully choosing the canonical points and fixing $\pi$ to $\overline{st}$,
so our variant of the oracle model \afterreview{is} slightly more general
(although practically equivalent).

We'd like a similar algorithmic property for the intersection graph in the homology cover $\overline G$
that we have for the auxiliary graph $H$.
The second construction will be what gives us that property,
by constructing $\overline G$ \emph{from} $H$:
\begin{itemize}
    \item For each arrangement vertex $v^b$ in $H$,
    let its adjacent vertices be denoted $x_1,\dots,x_k$.
    Delete $v^b$ and create a clique over $\{x_1,\dots,x_k\}$.
    \item After performing this for every arrangement vertex $v^b$,
    the result is \emph{exactly} the intersection graph in the homology cover $\overline G$.
    \item Therefore, the set of vertices visited by the shortest path from some $(\gamma,-)$ to $(\gamma,+)$ (over all possible $\gamma\in\cC$, breaking ties arbitrarily) corresponds to a minimum-weight set of obstacles separating $s$ and $t$.
\end{itemize}

For specific obstacle types,
we will usually work with $\overline G$,
although we will do it implicitly.
Hence, with the tools from this section, we can now discuss algorithms in all three model types
as solutions to a certain form of shortest-path problem.

\section{Algorithmic Results}
\label{sec:algorithms}

In this section, we will devise algorithms for the $(s,t)$ point-separation problem
based on our homology cover structures.
We will devise algorithms for all three model types:
Some that work with the arrangements of curves,
some that work with the ``oracle model'' (the intersection graph in the homology cover),
and some that work with specific types of obstacles.
Most of our algorithms are fairly simple reductions to various known shortest-path
algorithms, greatly simplifying some of the previous approaches to the $(s,t)$ point-separation problem.
However, some of them are more involved.

In the previous section, we assumed
no obstacle individually separated $s$ and $t$.
Before moving on, we quickly show this is enough:
\begin{observation}
Let $\cC$ be a set of obstacles,
and let $s$ and $t$ be points.
Suppose an oracle exists that,
for an obstacle $\c\in\cC$,
determines whether $\c$
itself separates $s$ and $t$,
all in $O(1)$ time.
Let $\cC_0$ be the set of all obstacles that do.
Then an instance of the weighted (unweighted) $(s,t)$ point-separation problem over $\cC$
can be reduced to an instance of the weighted (unweighted) $(s,t)$ point-separation
problem over $\cC\setminus\cC_0$ in $O(|\cC|)$ time.
In particular, such a reduction returns the best solution out of the solved subproblem (if any),
and each of the individual obstacles forming solutions (if any).
\end{observation}

\subsection{General Weighted Obstacles in Oracle Model}
With no significant further work,
we already obtain some na\"ive algorithms by using
\arxivappendix{\cref{lemma:sp}}
(or the constructions in \cref{sec:homology-informal}):
\begin{theorem}
\label{thm:match-cg16}
For points $s$ and $t$,
and a weighted set of obstacles $\cC$
that have $k$ pairwise intersections,
the $(s,t)$ point-separation problem
can be solved in $O(|\cC|\cdot k+|\cC|^2\log|\cC|)$ time.
\end{theorem}
\begin{proof}
Compute the intersection graph in the homology cover $\overline G$ of $\cC$.
For each $\c\in\cC$,
compute the \afterreview{shortest path} from $(\c,-)$ to $(\c,+)$ in $\overline G$.
The smallest such path induces the solution by
\arxivappendix{\cref{lemma:sp}}.
Each \afterreview{shortest path} can be computed in $O(k+|\cC|\log|\cC|)$ time by Dijkstra's algorithm with a Fibonacci heap~\cite{clrs},
and there are $O(|\cC|)$ total such paths.
\end{proof}
Note that \cref{thm:match-cg16} matches the result of \cite{CabelloG16}, although
the algorithm is much simpler.
When $k\in\Theta(|\cC|^2)$ (i.e., dense intersection graphs),
this algorithm runs in $O(|\cC|^3)$ time,
which could also be obtained by running Floyd-Warshall for APSP instead of Dijkstra.
It is not known \afterreview{if} the general form of weighted
APSP can be solved in ``truly'' subcubic time (see \arxivappendix{\cref{subsec:apsp}}).
However, since the weights are applied to vertices, not edges,
there is a known faster algorithm
(also discussed in \arxivappendix{\cref{subsec:apsp}}):

\begin{theorem}
\label{thm:better-than-n3}
For points $s$ and $t$,
and a weighted set of obstacles $\cC$,
the $(s,t)$ point-separation problem
can be solved in
$\widetilde{O}(|\cC|^{(3+\omega)/2})$
time in the oracle model,
where $2\leq\omega<2.371339$ is the matrix multiplication exponent.
\end{theorem}
\begin{proof}
Apply the same reduction to APSP as before, but use
the algorithm of
Abboud, Fischer, Jin, Williams, and Xi~\cite{abboud2025allpairs}
for APSP with real vertex weights.
\end{proof}

\subsection{General Unweighted Obstacles in Oracle Model}

We are also interested in the $(s,t)$ point-separation problem for unit weights,
which corresponds to \afterreview{shortest path}s for unit weights.
In this case, we can obtain faster algorithms:

\begin{theorem}
\label{thm:general-unweighted}
For points $s$ and $t$,
and an \emph{unweighted} set of obstacles $\cC$,
the $(s,t)$ point-separation problem
can be solved in $O(|\cC|^\omega\log|\cC|)$ time,
where $2\leq\omega<2.371339$ is the matrix multiplication exponent.
\end{theorem}
\begin{proof}
Compute the intersection graph in some homology cover $\overline G$ of $\cC$,
and then run unweighted undirected APSP~\cite{Seidel95}
on $\overline G$,
for $O(|\cC|^\omega\log|\cC|)$ total time.
\end{proof}

\subsection{General Weighted Obstacles in Arrangement Model}

We can also obtain results using the arrangement instead of the intersection graph.
The simplest of these is a slight modification of \cref{thm:match-cg16}:
\begin{theorem}
\label{thm:not-match-cg16-arrangement}
For points $s$ and $t$,
a weighted set of obstacles $\cC$,
and an arrangement $D,\sigma$ of $\cC$,
where $\sigma$ is given as lists of vertices,
let $m=\sum_{c\in\sigma}|c|$ be the total number of obstacle-vertex incidences.
Then, the $(s,t)$ point-separation problem
can be solved in $O(\min(|\cC|,|V(D)|)\cdot m+\min(|\cC|,|V(D)|)^2\log\min(|\cC|,|V(D)|))$ time.
\end{theorem}
\begin{proof}
Let $\overline H$ be the auxiliary graph in the homology cover,
which has $O(m)$ vertices and edges.
By \arxivappendix{\cref{lemma:sp}},
it suffices
to do one of the following:
\begin{itemize}
    \item Compute the single-source shortest-path (SSSP) tree from each obstacle in $\overline H$.
    \item Compute the SSSP tree from each arrangement vertex in $\overline H$.
\end{itemize}
Let $k=\min(|\cC|,|V(D)|)$.
Each SSSP computation can be performed in $O(m+k\log k)$ time,
and $k$ such computations suffice to solve the problem.
\end{proof}
This theorem in particular is of note because as we will see later,
it is essentially optimal assuming the APSP conjecture,
at least in the case when $|\cC|=\Theta(|V(D)|^2)$ and $m=\Theta(|\cC|)$,
as we will see in \cref{sec:lower-bounds}.

\subsection{Restricted Obstacle Classes without Weights}

As we will discuss in \arxivappendix{\cref{subsec:apsp}},
Chan and Skrepetos~\cite{Chan07,Chan10}
studied APSP for several forms of unweighted/undirected geometric intersection graphs.
Their intersection graphs are in the plane,
but with some extra work it is possible to study intersection graphs
in the homology cover, and consequently
the unweighted $(s,t)$ point-separation problem.
\begin{theorem}
\label{thm:restricted-si}
For points $s$ and $t$,
and an \emph{unweighted} set of obstacles $\cC$.
Let $\overline{\cC}$ be the set of $2|\cC|$ ``mapped'' obstacles in the homology cover.
Let $\si{n,m}$ (``static intersection'')
be the time complexity for checking if each of $n$ different
objects in $\overline{\cC}$ intersects any object in some subset $C\subset\overline{\cC}$
of size $m=|C|$.
Assume $\si{n,m}$ is super-additive, so that $\si{n_1,m_1}+\si{n_2,m_2}\leq\si{n_1+n_2,m_1+m_2}$.
Then $(s,t)$ point-separation over $\cC$ can be solved in $O(n\si{n,n})$ time.
\end{theorem}
\begin{proof}
Run APSP for the intersection graph in the homology cover via
the algorithm of Chan~\cite[Theorem 2]{ChanS19}
(see further discussion in
\arxivappendix{\cref{subsec:apsp}}).
\end{proof}

To use this result, we need efficient algorithms for static intersection
in the homology cover:

\begin{lemma}
\label{lemma:si-values-homology-cover}
For points $s$ and $t$,
the following values of $\si{n,n}$ hold for restricted obstacle types in the homology cover:
\begin{center}
    \begin{tabular}{l|c|l}
        Obstacle Class & $\si{n,n}$ \\
        \hline
        General Disks & $O(n \log n)$ \\
        Axis-Aligned Line Segments & $O(n \log \log n)$ \\
        Arbitrary Line Segments & $O(n^{4/3} \log^{1/3} n)$ \\
    \end{tabular}
\end{center}
\end{lemma}
The proof of the lemma is left to \arxivappendix{\cref{sec:static-intersection-hc}}.
At a high-level,
all cases are first reduced to the planar static intersection problem.
In the case of line segments, this becomes the standard planar static intersection problem
for line segments.
For disks, this is not the case, and the algorithm is more involved.

The combination of these two results give an important corollary:
\begin{corollary}
\label{cor:restricted-si-actual}
For points $s$ and $t$,
the unweighted $(s,t)$ point-separation problem
can be solved in $O(n^2\log\log n)$ time for axis-aligned line segments (or $O(1)$-length rectilinear polylines),
$O(n^{7/3}\log^{1/3}n)$ time for line segments (or $O(1)$-length polylines),
and $O(n^2\log n)$ time for general disks
or circles that do not contain both $s$ and $t$.
\end{corollary}

\subsection{Restricted Obstacle Classes with Weights}

We will now present a method for solving the weighted $(s,t)$ point-separation problem
using a tool called ``biclique covers'',
which are essentially a tool for a type of graph sparsification.

For a graph $G=(V,E)$,
a \defn{biclique} in $G$ is a complete bipartite subgraph ($A\times B\subset E$, where $A,B\subset V$ are disjoint).
The \defn{size} of a biclique is the number of vertices it contains ($|A|+|B|$).
A \defn{biclique cover} is a collection of bicliques in $G$ covering the edges $E$,
and its \defn{size} is the the sum of all sizes in its bicliques.
Biclique covers of geometric intersection graphs in two-dimensions are well-studied.
We summarize known results in \cref{tab:biclique-cover-results}.

\begin{table}[h]
\centering
\begin{tabular}{|l|l|l|l|}
\hline
\textbf{Graph Type} & \textbf{Cover Size} & \textbf{Construction Time} \\ %
\hline
line segment intersection & $\widetilde{O}(n^{4/3})$ & $\widetilde{O}(n^{4/3})$ \\ %
\hline
Axis-aligned line seg.~intersection & $\widetilde{O}(n)$ & $\widetilde{O}(n)$ \\ %
\hline
$k$-clique-free line seg.~intersection & $\widetilde{O}_k(n)$ & $\widetilde{O}_k(n)$ \\ %
\hline
\end{tabular}
\caption{Known results for biclique covers of 2D geometric intersection graphs.
All of these results are stated and proven by Chan~\cite{chan2023finding},
although essentially all of the steps have appeared in a number of prior works.
The notation $\widetilde{O}$ hides logarithmic factors. The notation $\widetilde{O}_k$ further assumes that $k$ is constant. Note that axis-aligned line segment intersection graphs are $K_3$-free.
}
\label{tab:biclique-cover-results}
\end{table}

Biclique covers are useful in our case because they can be used for faster \emph{vertex-weighted} \afterreview{shortest paths}.
In particular:
\begin{lemma}
\label{lemma:biclique-cover-to-vw-apsp}
Let $G=(V,E)$ be an $n$-vertex (undirected) graph with vertex-weights
admitting a biclique cover of size $S(n)$ that can be constructed in $T(n)$ time.
Then APSP over $G$ can be solved in $O(T(n)+n\cdot S(n)\log(n))$ time.
\end{lemma}

We leave the proof to \arxivappendix{\cref{sec:appendix-biclique}}.

These results aren't quite enough for the $(s,t)$ point-separation problem,
since what we need is actually a biclique cover \emph{in the homology cover}.
Fortunately, we can construct this:

\begin{lemma}
\label{lemma:biclique-covers-homology-cover}
Let $s$ and $t$ be designated points in the plane,
and let $\cC$ be a set of line segments with $n=|\cC|$.
Let $\overline G$ be the intersection graph in the homology cover,
and let $G$ be the intersection graph in the plane.
Then $\overline G$ has a homology cover of size $\widetilde{O}(n^{4/3})$
that can be found in $\widetilde{O}(n^{4/3})$ time.
Moreover, if $G$ contains no $k$-clique
(including if $\cC$ is a set of rectilinear segments,
in which case $G$ contains no $3$-clique),
then $\overline G$ has a homology cover of size $\widetilde{O}_k(n)$
that can be found in $\widetilde{O}_k(n)$ time.
\end{lemma}

We leave the proof of this to
\arxivappendix{\cref{sec:appendix-biclique}}
as well.
The proof is similar to that of
\cref{thm:restricted-si}.

The combination of these results gives the following theorem:

\begin{theorem}
\label{cor:separation-biclique-cover-times}
For a weighted set of obstacles $\cC$,
and points $s,t$,
the $(s,t)$ point-separation can be solved in $\widetilde{O}(n^{7/3})$ time
if $\cC$ is a set of line segments,
$\widetilde{O}(n^2)$ time if $\cC$ is a set of axis-aligned line segments,
and $\widetilde{O}_k(n^2)$ time if $\cC$ is a set of line segments whose
intersection graph has no $k$-clique.
The same bounds hold if each obstacle is an $O(1)$-length polyline
among lines of the same properties.
\end{theorem}

Note that a biclique cover for $O(1)$-length polylines
can be recovered from biclique covers over the individual line segments
(adjusted slightly so that the line segments in the same polyline do not overlap).

\section{Lower Bounds}
\label{sec:lower-bounds}

In this section, we present several related fine-grained lower bounds
for specific cases of the $(s,t)$ point-separation problem.
The main intermediate tool is the problem of finding a minimum-weight walk of length $k$ in a directed graph,
for a given $k$ (or detecting if any walk of length $k$ exists).
All of our lower bounds are based on the following unified theorem:
\begin{theorem}
\label{thm:k-cycle-reduction}
For a positively edge-weighted directed graph $G=(V,E)$ with $n$ vertices and $m$ edges with maximum edge-weight $W$, and an integer $k$,
there exists three sets $\cC_1,\cC_2,\cC_3$ of $2km+6m$ obstacles,
each with a weight equal to that of some edge in $G$, so that each of the following properties holds for one set:
\begin{itemize}
    \item All obstacles are line segments.
    \item All obstacles are length-$2$ polylines, and the total number of unique intersection points of the obstacles is $k+2(k+1)m+4m=\Theta(km)$.
    \item All obstacles are length-$3$ rectilinear polylines.
\end{itemize}
Moreover, there are points $s$ and $t$ in the plane so that
$G$ has a walk of length $k$ with weight at most $w$
if and only if there is some subset of $\cC_i$ (for any $i\in\{1,2,3\}$)
that separates $s$ and $t$ and has weight at most $w+(k+6)W$,
so long as $w\leq kW$.
Furthermore, each of $\cC_1,\cC_2,\cC_3$ can be constructed in time proportional to their sizes.
\end{theorem}

\begin{figure}
\centering
\includegraphics[scale=0.4,page=8]{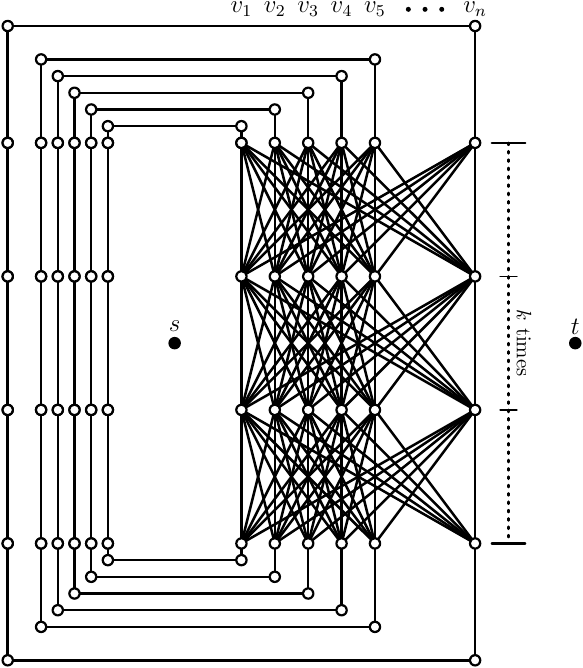}
\hspace{4em}
\includegraphics[scale=0.4,page=9]{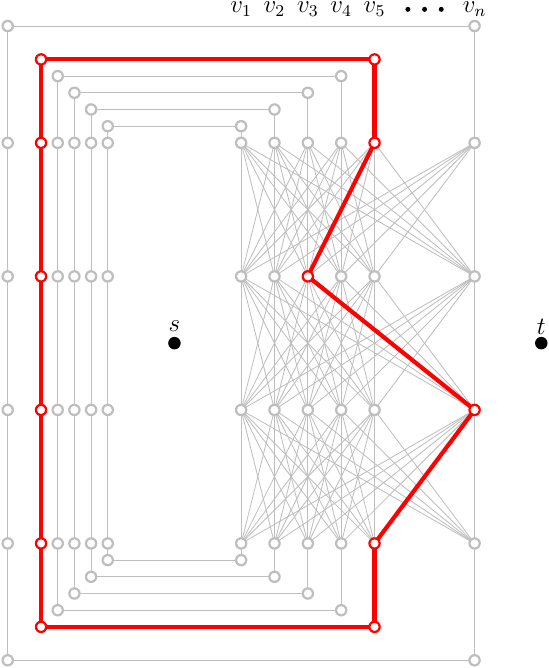}
\caption{A demonstration of the sub-cubic reduction from min-weight $k$-walk to the $(s,t)$ point-separation problem with line segments. Empty circles are used to annotate the ends segments.}
\label{fig:lower-bound-segments}
\end{figure}

\begin{proof}
For each property, we will construct a reduction
from min-weight $k$-walk (for fixed $k$)
in a directed edge-weighted graph to point separation,
with the stated guarantees.
The constructions for each property are essentially the same.
We will focus on the line segment case, and explain the modifications afterwards.

The construction is visualized in \cref{fig:lower-bound-segments}.
For each vertex $v_i$ ($i\in\{1,\dots,n\}$),
assign the column $x=i$.
For each value $r\in[k+1]$,
assign the row $y=r$.
For each directed edge $(v_i,v_j)\in E$, with weight $w_{v_i,v_j}$
create $k$ line segments,
each also with weight $w_{v_i,v_j}$.
Each should go from the $i$th column to the $j$th column,
and the $r$th row to the $(r+1)$th row (for each $r\in[k+1]$).
We call these the \defn{edge segments}.
Pick $s$ to be some point to the left of the columns, and $t$ to be some point to the right.
For the $i$th column, create a set of $k+6$ rectilinear line segments
connecting the top point of the $i$th column to the bottom,
so that the path formed by these line segments lies to the left of $s$ at its $x$ coordinate,
ensuring that these line segments do not intersect the corresponding line segments generated for other columns,
nor do they cross the minimal orthogonal rectangle containing the previous set of line segments.
We call these the \defn{vertex segments}.
These two segment types form the entirety of the obstacles,
so the stated time complexity follows.
We need only prove correctness of the reduction.

A $k$-walk through $G$ starting and ending at a vertex $v_i\in V$
corresponds exactly to a connected set of edge segments from the point $(i,0)$ to $(i,k)$,
and both also have the same weight.
The forward direction of the reduction follows: Given such a set of edge segments,
a set of vertex segments of weight exactly $(k+6)W$ exists (those for the $i$th column)
so that the union of the two separates $s$ and $t$,
and has the specified weight.
For the backwards direction of the reduction:
For a set $C\subset\cC$ with weight at most $w+(k+6)W$, if $w\leq kW$,
then it is only possible to have one maximal connected set of vertex segments
in $C$.
Each of the sets of (potential) obstacles crossed by paths in
\cref{fig:lower-bound-segments-paths}
must have at least one obstacle chosen from them for the whole set to separate $s$ and $t$,
so at least one such maximal connected set of vertex segments must be in $C$.
Hence, the total weight of the remaining obstacles in $C$ is $w$,
and some subset of these obstacles must themselves be a connected set of edge segments
(one per row) corresponding exactly to a $k$-walk in $G$.

The two modified constructions with polylines are obtained by replacing the edge segments with polylines
that have the desired properties.
These cases are visualized in
\cref{fig:lower-bound-polylines}.
\end{proof}

\begin{figure}
\centering
\includegraphics[scale=0.35,page=12]{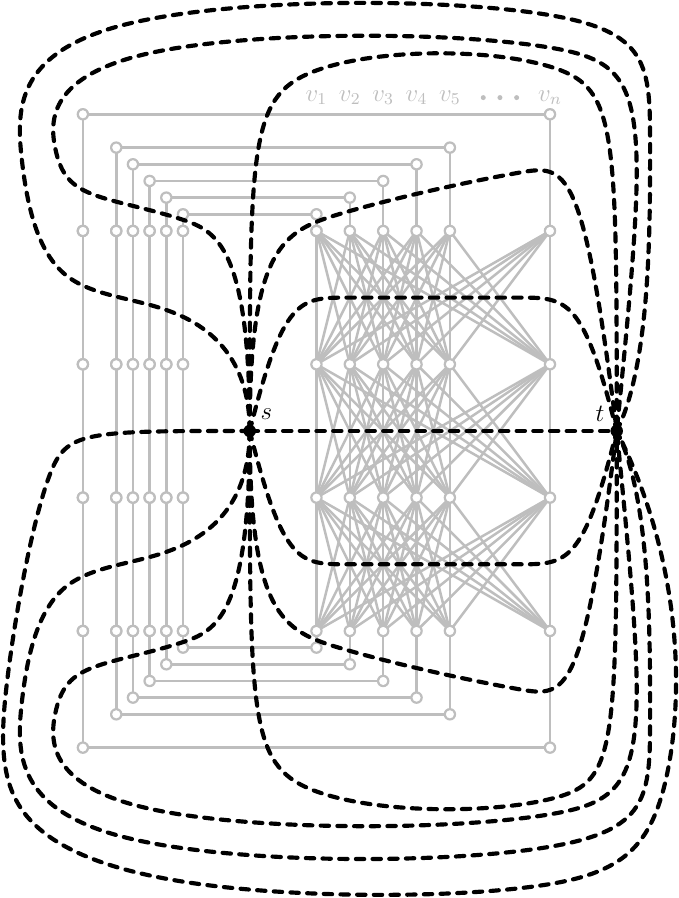}
\caption{The finite set of $(s,t)$ paths that need to be considered for the reduction of min-weight $k$-walk to the $(s,t)$ point-separation problem with line segments.}
\label{fig:lower-bound-segments-paths}
\end{figure}

\begin{figure}
\centering
\includegraphics[scale=0.4,page=10]{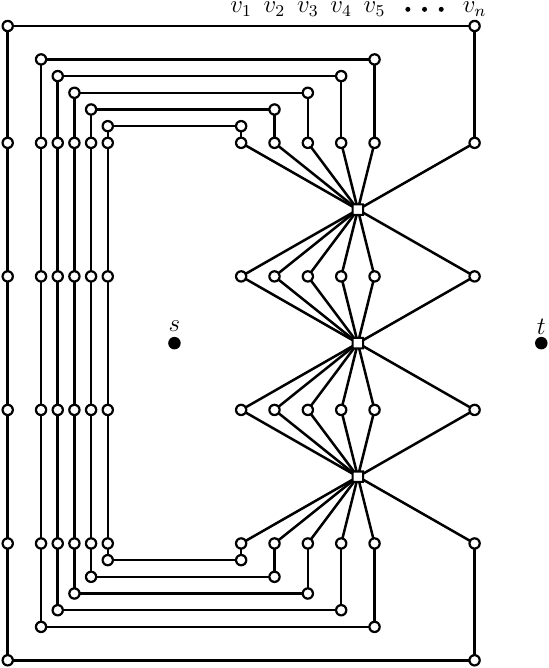}
\hspace{2em}
\includegraphics[scale=0.4,page=11]{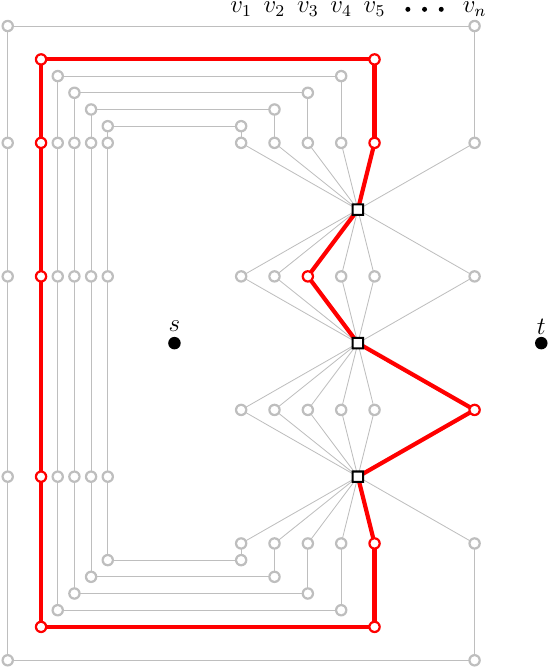}\\
\vspace{1em}
\includegraphics[scale=0.4,page=13]{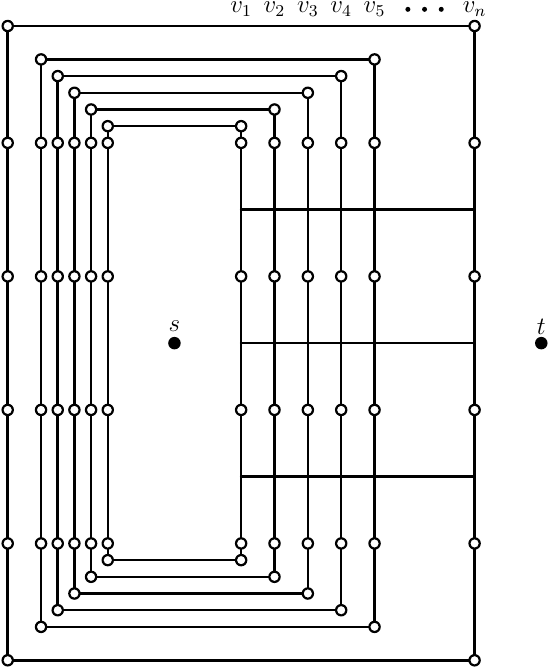}
\hspace{2em}
\includegraphics[scale=0.4,page=14]{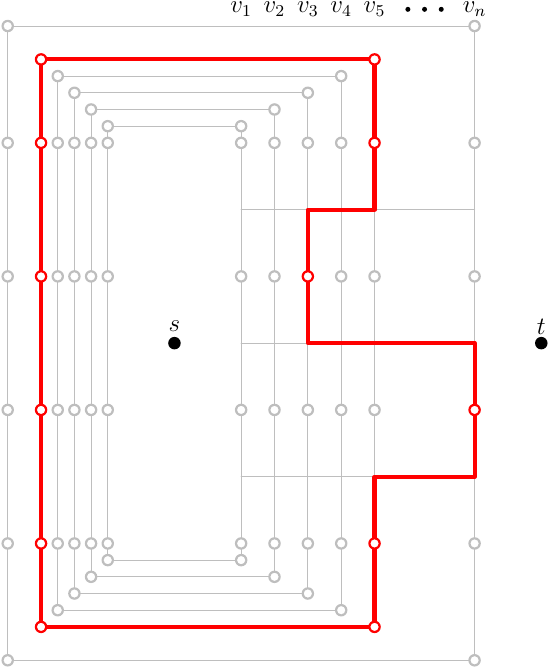}
\caption{A demonstration of the reduction from min-weight $k$-walk to the $(s,t)$ point-separation problem for
length-$2$ polylines with few total intersection points,
or length-$3$-rectilinear polylines.}
\label{fig:lower-bound-polylines}
\end{figure}

There are a number of fine-grained lower bounds for min-weight $k$-walk and $k$-walk detection.
In particular, all such bounds hold for \emph{fixed} $k$ and graphs where $k$-walks and $k$-cycles coincide
(which is itself always true for $k=3$ in graphs with no self-loops, and is otherwise implied by a property called ``$k$-circle-layered'').
In \arxivappendix{\cref{sec:lb-details}},
we discuss the following results that follow from known fine-grained bounds and 
\cref{thm:k-cycle-reduction}:
\begin{tabular}{lll}
\toprule
\begin{socgtabular}{@{}l@{}}
If $(s,t)$ point-separation\\ can be solved in time\dots
\end{socgtabular}
&
\begin{socgtabular}{@{}l@{}}
for $n$ obstacles that are\dots
\end{socgtabular}
&
\begin{socgtabular}{@{}l@{}}
then it would\\ imply a new \afterreview{state-of-the-art}\\ algorithm for\dots
\end{socgtabular}
\\
\midrule
$O(n^{3/2-o(1)})$ & weighted line segments & edge-weighted APSP\\
$O(n^{2-o(1)})$ & weighted line segments & minimum-weight $k$-clique\\
$O(n^{2-o(1)})$ & weighted length-$3$ rectilinear polylines & minimum-weight $k$-clique\\
$O(n^{3/2-o(1)})$ & unweighted line segments & max-$3$-SAT\\
$O(n^{3/2-o(1)})$ & unweighted length-$3$ rectilinear polylines & max-$3$-SAT\\[4pt]
$O(n^{3/2-o(1)})$ &
\begin{socgtabular}{@{}l@{}}
weighted length-$2$ polylines with $O(\sqrt{n})$\\ unique intersection points and $3$\\ intersection points per obstacle
\end{socgtabular}
& edge-weighted APSP\\
\bottomrule
\end{tabular}

\section{Conclusion}

We have discussed several upper and lower bounds for the $(s,t)$ point-separation problem in various cases,
and we conclude by briefly listing some of the most interesting open problems:
\begin{itemize}
    \item Is there a near-quadratic algorithm for general weighted obstacles? Can a conditional lower bound (based on any popular hypotheses) excluding this possibility be formulated?
    \item Can a (truly) sub-quadratic algorithm be devised for disks, unit disks, or axis-aligned segments? Is there a non-trivial fine-grained lower bound in any of these cases?
    \item Are there faster algorithms or stronger lower bounds for the \emph{unweighted} versions of the problem?
\end{itemize}
We note that it seems very likely that unit disks would admit a \emph{slightly} sub-quadratic algorithm,
since APSP is known to be solvable for unit disk graphs in (essentially) slightly sub-quadratic time~\cite{chan2016all},
and the methods seem as though this method could plausibly be extended to the homology cover.

\bibliography{references}

\arxivonly{
\appendix

\section{Prior Work}
\label{sec:prior-work}

Several works have studied the $(s,t)$ point-separation problem.
In addition, a few generalizations of the $(s,t)$ point separation have been studied.
The $(s,t)$ point separation problem
also itself generalizes the well-studied problem of maximum flow in planar graphs.
We review all these results in this section, as well as results for all-pairs \afterreview{shortest paths}, which will be quite relevant to our methods.

\subsection{$(s,t)$ Point-Separation}

Various results are known in each model.
The first model considered was
the restriction to unit disk obstacles
(equivalently, unit circle obstacles that do not contain $s$ or $t$ in their interior).
For the special case of (unweighted) unit disk obstacles,
Gibson, Kanade, and Varadarajan~\cite{GibsonKV11}
gave a polynomial-time $2$-approximation for this problem.
Cabello and Milinkovi\'{c}~\cite{CabelloM18} later described an algorithm
running in $O(|\cC|^2\log^3|\cC|)$ time solving this case exactly.

Cabello and Giannopoulos~\cite{CabelloG16}
proposed the oracle-based intersection graph model.
Their model encapsulates many important classes of curves (with and without interiors),
including disks and circles (of any radii), line segments, certain algebraic curve segments,
as well as various combinations of these classes.
In particular, the purpose of their model is to allow them to compute an
intersection graph
along with some additional information related to $\overline{st}$.
The additional information they store is inherently \emph{topological},
as are important aspects of all algorithms (existing and new)
that we will discuss in detail throughout this work.

In the oracle-based intersection graph model,
Cabello and Giannopoulos~\cite{CabelloG16}
present an algorithm that solves the $(s,t)$ point-separation problem in $O(|\cC|^3)$ time
for general weights.
To be more precise, when the intersection graph has exactly $r$ edges (that is, $r$ pairs of obstacles intersect),
their algorithm runs in
$O(|\cC|\cdot r+|\cC|^2\log|\cC|)$ time.
The algorithm of Cabello and Milinkovi\'{c}~\cite{CabelloM18} for unit disks builds upon this approach
by leveraging the structure of unit disks.

Kumar, Lokshtanov, Saurabh and Suri~\cite[Section 6]{KumarLSS21}
later described an algorithm for the arrangement-based model.
We will use the following formulation and notation for the arrangement-based model throughout the paper:
Given an embedded plane graph $D$, a (weighted) set of connected subgraphs $\sigma$,
and two faces $s$ and $t$,
the \defn{$(s,t)$ face-separation problem} asks for a minimum-weight
subset of $\sigma$ whose edges
form an $(s,t)$ cut in the dual graph to $D$.
It is well-known that this corresponds to a
subgraph of $D$
containing a simple cycle separating $s$ and $t$~\cite{itai1979maximum}
(in the same sense as before).
Note that
the $(s,t)$ point-separation problem
can be reduced to the $(s,t)$ face-separation problem by
taking the \defn{arrangement} of the obstacle set $\cC$.
That is, the plane graph $D$ of minimal vertex count
for which the union of all obstacles in $\cC$ is a planar drawing,
and the set $\sigma$ of edge sets corresponding to arcs of each obstacle in $\cC$.
We carefully define it in this way because we have made \emph{no assumptions} of general position:
Pairs of obstacles may share more than a finite number of points.
For many important classes of closed curves
(same list as before)
the arrangement has $O(|\cC|^2)$ vertices,
and possibly far fewer for certain instance sets.

By assuming the arrangement $D$ and its corresponding connected subgraphs $\sigma$ are given as input,
Kumar, Lokshtanov, Saurabh and Suri~\cite[Section 6]{KumarLSS21}
describe an algorithm that runs in polynomial time.
Narrowing down the exact time complexity of their algorithm actually requires defining another parameter set.
For each vertex $v\in V(D)$, let $p_v$ denote
the number of connected subgraphs in $\sigma$ that
contain $v$.
Then, their algorithm runs in
$O(|V(D)|\cdot(\sum_{v\in V}p_v^2+\sum_{S\in\sigma}|S|)\cdot\log|V(D)|)
=O(|V(D)|\cdot(\sum_{v\in V}p_v^2)\cdot\log|V(D)|)$
time.
In less precise terms, if the total number of vertex-obstacle incidences
is $m$,
the total runtime is bounded by
$O(|V(D)|\cdot m^2\cdot\log|V(D)|)$.
We will not be giving a full summary of their algorithm,
but we note that
their algorithm implicitly makes use of one key topological construct (perhaps unintentionally)
that turns out to be quite helpful.
Specifically, they implicitly work with a structure known as a ``homology cover'',
although they do not refer to it as such.
Homology covers (of this particular flavour)
have been previously applied to the maximum flow problem in so-called surface graphs~\cite{EricksonN11a,ChambersEFN23}.
In our work, we will also (explicitly) make use of the homology cover.

\subsection{Maximum Flow}
Consider the special case of the $(s,t)$ point-separation problem
where the obstacles are a set of non-crossing line segments
that may share endpoints.
In such a case, the line segments form the edges of a planar graph,
and each of $s$ and $t$ belong to some face,
and the $(s,t)$ point-separation problem is equivalent to the min $(s,t)$-cut problem in the dual graph.
By linear programming duality,
this is equivalent to the maximum flow problem,
and hence can be solved in polynomial time.
In particular, since the graph is planar,
a line of work~\cite{itai1979maximum,reif1983minimum,Frederickson87,ItalianoNSW11}
has given particularly fast algorithms for maximum flow in planar graphs,
resulting in an $O(n\log\log n)$ time for a graph of $n$ vertices~\cite{ItalianoNSW11},
assuming the graph itself is given.
This entire line of work primarily focuses directly on the min $(s,t)$-cut problem,
since it turns out to be easier to work with.
In particular, the edges of a min $(s,t)$-cut in a planar graph always form a \emph{simple cycle} in the dual graph.

A related sequence of work has also obtained similar results
for maximum flow in graphs embedded on a surface of
bounded genus~\cite{ImaiI90,ChambersEN09,ChambersEN09_2,ChambersEN12,EricksonN11a,ChambersEFN23},
the most recent paper of which obtains algorithms running in $O(c^{O(c)}n\log\log n)$ and $O(2^{O(c)}n\log n)$ time
for $c$ which is the sum of the genus and the number of boundaries of the surface.
In particular, the latter algorithm
works with the ``homology cover'',
and is quite similar to the algorithm of \cite{KumarLSS21}.
We will not discuss in further detail how to approach such problems on surfaces
(\cite{ChambersEFN23} presents a fantastic introduction for the curious reader),
but they are quite related to the topological methods used for the $(s,t)$-point separation problem.

\subsection{Generalizations of $(s,t)$ Point-Separation}

There is a natural generalization of the $(s,t)$ point-separation problem
that has been explored:
Instead of just one pair $(s,t)$,
in the \defn{generalized point-separation problem}
we are given pairs $(s_1,t_1),\dots,(s_p,t_p)$,
all of which must be separated by the subset of obstacles.
In the model of \cite{CabelloG16},
Kumar, Lokshtanov, Saurabh, Suri, and Xue~\cite{KumarSoCG2022} gave an algorithm for this problem
running in $O(2^{(O(p)}n^{O(|\cC|)})$ time.
Assuming the exponential time hypothesis,
they also showed that
this problem cannot be solved
in time $O(f(p)n^{o(|\cC|/\log|\cC|})$,
so their algorithm is essentially optimal.
It should be noted that this problem can also be framed in terms of a (general) graph $G$,
where $V(G)=\{s_1,\dots,s_p,t_1,\dots,t_p\}$ (note that some of these may coincide)
and $E(G)=\{s_1t_1,\dots,s_pt_p\}$.
Their lower bound also holds even in the case that such a graph $G$ is a complete graph.
Chan, He, and Xue~\cite{ChanH024}
gave two different polynomial-time approximation algorithms for the case when
$G$ is a star whose centre is the point at infinity (equivalently, a sufficiently distant point).

Before the results on the general case,
Gibson, Kanade, Penninger, Varadarajan, and Vigan~\cite{GibsonKPVV16}
considered the generalized point-separation problem
with $G$ as the complete graph
and (both general and unit) disk obstacles.
They showed that the problem is NP-hard even with just unit disks,
and that it admits a polynomial-time $(9+\varepsilon)$-approximation algorithm
for any $\varepsilon$ and general disks.

\subsection{All-Pairs Shortest Paths}
\label{subsec:apsp}

The all-pairs shortest paths problem (APSP)
is a fundamental problem in graph algorithms
with extensive literature,
and (as we will see) it is quite related to the $(s,t)$ point-separation problem.
In fact, there are three distinct variants that will be of importance.

For a weighted directed graph $G$ with $n$ vertices and real-valued edge weights,
APSP is known to be solvable in $O(n^3)$
time by the Floyd-Warshall algorihm~\cite{clrs}.
It is also known to be solvable in $O(n^3/2^{\Omega(\sqrt{\log n})})$ time~\cite{Williams14a, ChanW16, Williams18, ChanW21}.
However, there is no known algorithm solving APSP for real-valued weights in $O(n^{3-\varepsilon})$ time for any $\varepsilon>0$, and the existence of one has become a long-standing open problem.
This has lead to the so-called \defn{APSP-conjecture}, which states that no such algorithm exists.
It turns out that, similarly to \textsc{3-SAT}, \textsc{3-SUM}, and the orthogonal vectors problem,
there is a whole family of problems for which this conjecture has an equivalent form~\cite{WilliamsW10, WilliamsW18}.
This family uses \defn{subcubic reductions}:
A proof that if a particular problem \textsc{A} can be solved in $O(n^{3-\varepsilon})$ time for any value $\varepsilon>0$,
then there exists some value $\varepsilon'>0$ such that a problem \textsc{B} can be solved in $O(n^{3-\varepsilon'})$ time.
We will use the typical definitions of \defn{APSP-complete} and \defn{APSP-hard}
for problems in this family, and problems at least as hard as any problem in this family, repsectively.
APSP-complete problems typically have some form of matrix representation,
such as the (weighted) adjacency matrix of a graph.
For this reason, the term ``cubic'' here is a bit of a misnomer.
The size of the input for a typical problem in this family is $N=\Theta(n^2)$,
and so a ``subcubic'' algorithm in this case is one that runs in $N^{3/2-\varepsilon}$ time
for a value $\varepsilon>0$.
This is not a concern for most APSP-complete problems,
since a matrix representation is usually the simplest.
However, this will be important to us
while studying lower bounds for the $(s,t)$ point-separation problem.

In contrast to APSP over directed and edge-weighted graphs,
APSP on an undirected and unweighted graph is known to be solvable in $O(n^\omega\log n)$ time~\cite{Seidel95},
where $\omega<2.371339$ is the matrix multiplication exponent~\cite{alman2024asymmetryyieldsfastermatrix}.
Additionally, even faster algorithms are known for shortest paths in some intersection graphs:
Chan and Skrepetos~\cite{ChanS17, ChanS19}
give a framework reducing APSP in (unweighted) intersection graphs
to an offline intersection detection problem.
Importantly, their framework is \emph{not} restricted to intersection graphs in the plane.
Their framework uses the following theorem:
\begin{theorem}[{\cite[Theorem 2]{ChanS19}}]
\label{thm:unweighted-intersection-apsp}
Let $\cC$ be a set of geometric objects in some space.
Let $\si{n,m}$ (``static intersection'') be the time complexity for
checking if each of $n$ different
objects in $\cC$ intersect
any object in some subset $C\subset\cC$ of size $m=|C|$.
Assume $\si{n,m}$ is super-additive, so that $\si{n_1,m_1}+\si{n_2,m_2}\leq\si{n_1+n_2,m_1+m_2}$.
Then APSP in the unweighted intersection graph
of $\cC$
can be solved in $O(n^2+n\si{n,n})$ time.
\end{theorem}
In particular, Chan and Skrepetos~\cite{ChanS17, ChanS19}
also give or cite such data structures for several classes of obstacles.
We summarize relevant values of $\si{n,n}$ in \cref{fig:si-values}.
In particular, note that the result for arbitrary line segments
also extends to polylines of length $O(1)$,
and the result for axis-aligned line segments
extends to orthogonal polylines of length $O(1)$.

\begin{table}[h]
    \centering
    \begin{tabular}{l|c|l}
        Obstacle Class & $\si{n,n}$ & Citation \\
        \hline
        General Disks & $O(n \log n)$ & \cite{ChanS17, ChanS19} \\
        Axis-Aligned Line Segments & $O(n \log \log n)$ & \cite{ChanS17, ChanS19} \\
        Arbitrary Line Segments & $O(n^{4/3} \log^{1/3} n)$ & \cite[Theorem 4.4]{Chazelle93a} \\
    \end{tabular}
    \caption{Summary of $\si{n,n}$ values for various obstacle classes in the plane.}
    \label{fig:si-values}
\end{table}

Lastly, there are some faster algorithms known for APSP over \emph{vertex}-weighted graphs,
which surprisingly seems to be easier than APSP over edge-weighted graphs.
Specifically,
it is known that APSP with vertex weights can be solved in
$\widetilde O(n^{3-(3-\omega)/4})$ time,
where the $\widetilde O$ notation hides logarithmic factors, and $\omega$ is the matrix-multiplication exponent~\cite{Chan07,Chan10}.
Let $\omega(a,b,c)$ be the value so that an
$n^a\times n^b$ matrix and a $n^b\times n^c$ matrix
can be multiplied in $O(n^\omega(a,b,c))$ time.
It is also known that APSP with vertex weights can be solved in
a similar (and slightly faster) time complexity using \emph{rectangular} matrix multiplication~\cite{Yuster09}.
Very recently, the running time for vertex-weighted APSP has been improved to $\widetilde{O}(n^{(3+\omega)/2})$ time~\cite{abboud2025allpairs}.

\section{Homology and Obstacles}
\label{sec:homology-formal}

In this section, we primarily review aspects of homology, and their interactions with obstacle curves.
Throughout this section, we will assume that no single obstacle $\c\in\cC$
separates $s$ and $t$.
We will briefly revisit and justify this assumption later while describing algorithms.

\subsection{Homology in the Plane}
\label{subsec:homology}

Homology is a broad field of study in algebraic topology.
Fortunately, we need only limit ourselves to a very small special case.
Our overview here is greatly simplified and intended to be approachable for non-experts.%

For two points $s$ and $t$, consider simple closed curves $c,c'$ in $\plane\setminus\{s,t\}$,
where $\plane$ denotes the \emph{extended} plane with a point at infinity (so that $\plane$ is homeomorphic to a sphere),
so that $\plane\setminus\{s,t\}$ is homeomorphic to an annulus.
We will often refer to curves in $\R^2\setminus\{s,t\}$, but for topological purposes
these should be considered to be curves in $\plane\setminus\{s,t\}$ (a superset).
In the $(s,t)$ point-separation problem, we are essentially studying covers of closed curves
that separate $s$ and $t$.
We will say $c$ and $c'$ are \defn{homologous}
if either they \emph{both} separate $s$ and $t$,
or they \emph{both} do not separate $s$ and $t$.
This is an equivalence relation,
and hence defines equivalence classes that we refer to as \defn{homology classes}.
These definitions coincides with the topological notion of ``one-dimensional $Z_2$-homology classes in the annulus''.
This is the \emph{only} notion of homology we will use in this work, so we use the simplified terms.
\cite{ChambersEFN23} give a more detailed (and still approachable) explanation overview of one-dimensional $Z_2$-homology in other surfaces.
We can obtain a very important fact from the field of
homology:
\begin{fact}
\label{fact:homology-path-crossings}
Let $\pi$ be a fixed simple path in the plane from $s$ to $t$.
Let $c$ be a simple closed curve in $\plane\setminus\{s,t\}$.
Then the homology class of $c$
is defined by its number of crossings%
\footnote{Some care must be taken for the definition of ``crossings'' in this context:
An orientation is applied to the curve $\pi$, so that it has a ``right'' and ''left'' side,
and any open set $S$ of $\plane\setminus\{s,t\}$ where $S\setminus\pi$ has exactly two connected
components can have ``left'' and ''right'' assigned to those components.
The curve itself is ``assigned'' to the right side.
That is, a crossing is a point $x$ in the intersection of $c$ and $\pi$,
and a ``direction'' along $c$,
so that
for every $\varepsilon>0$,
there is some open interval along $c$ whose left-endpoint is $x$ (resp. right-endpoint, depending on ``direction'')
contained in some open set $S$ such that $S\setminus\pi$ has two connected components,
and $S$ itself is contained in $N_{\varepsilon}(x)$.
Note that a point $x$ may be part of either $0$, $1$, or $2$ crossings.}
with $\pi$, modulo $2$.
Moreover, the homology class of simple closed curves separating $s$ and $t$
is exactly the one with $1$ crossing of $\pi$, modulo $2$.
Furthermore, the resulting classes are independent of the choice of $\pi$.
\end{fact}
See \cref{fig:homology-examples-1} for examples.

\begin{figure}
\centering
\includegraphics[scale=1.0,page=1]{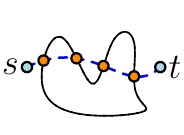}
\hspace{3em}
\includegraphics[scale=1.0,page=5]{homology-examples}
\caption{Examples of two curves with different homology classes.
}
\label{fig:homology-examples-1}
\end{figure}

This is a very important result for our algorithmic problem.
With this, we have a straightforward way of testing
if a curve separates $s$ and $t$.
In fact, the resulting algorithm (counting crossings)
is quite similar to a folklore point-in-polygon algorithm
that counts crossings of the polygon with any ray originating from the query point.

This allows us to define a very helpful structure:
Create an incision (boundary) in the plane along a simple path $\pi$ from $s$ to $t$
(i.e. delete the entire path).
Then, create a copy of the plane with $\pi$ removed.
Create two new copies of $\pi$ (with the endpoints $s$ and $t$ removed).
For each copy of $\pi$, connect its top side to one plane, and its bottom side to the other.%
\footnote{We omit a formalization this construction, but in short:
One could define additional generating open sets from neighbourhoods in the original plane
that are sufficiently small so that each has exactly one entry/exit point of $\pi$,
and hence can be lifted into the new space using $\pi$ as a reference.
Simpler specialized forms like those used by \cite{ChambersEFN23} would also suffice
for our purposes later.}
The result is a space we call the \defn{homology cover} of the plane
w.r.t. the path $\pi$.%
A point in the homology cover is defined by its projection into $\R^2$,
as well as a value $\{0,1\}$ indicating which copy it belongs to.
We obtain the following lemma:
\begin{lemma}[\cite{ChambersEFN23}]
\label{lemma:homology-cover-correctness}
Let $p=(p',b)\in(\R^2\setminus\{s,t\})\times\{0,1\}$ be a point in the homology cover of the plane
with respect to an $s-t$ path $\pi$.
Let $c^*$ be a path in the homology cover
from $p$ to $(p',1-b)$.
Let $c$ be the projection of $c^*$ into $\R^2$.
Then $c$ separates $s$ and $t$.
Moreover, a simple closed curve $c$ in $\R^2\setminus\{s,t\}$
separates $s$ and $t$ if and only if some path $c^*$ with this property exists.
\end{lemma}
This result implied by standard properties of homology,
and we will make heavy use of it.
Chambers et al.~\cite{ChambersEFN23} present a slightly more general version
in an approachable graph-theoretic form.
See \cref{fig:homology-examples-2} for an example of this lemma,
where
the ``starting point'' corresponds to $p$ or $p'$.
An equivalent form of this lemma would be to observe
that a connected curve can either map to one or two connected components
in the homology cover, depending on the homology class of the curve.

\begin{figure}
\centering
\begin{minipage}{0.45\textwidth}
  \centering
  \begin{tikzpicture}[node distance=1cm, every node/.style={inner sep=0pt}]
    \node (base) {\includegraphics[scale=1.0,page=2]{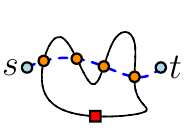}};
    \node (cover) [below=1cm of base] {%
      \fbox{%
        \begin{minipage}{\linewidth}
          \centering
          \includegraphics[scale=1.0,page=3]{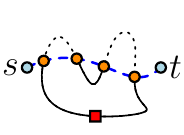}\\
          \frequentedotfill\\
          \includegraphics[scale=1.0,page=4]{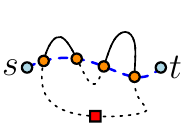}
        \end{minipage}%
      }%
    };
    \draw[|-,decoration={markings,mark=at position 1 with
    {\arrow[scale=2,>=stealth]{>}}},postaction={decorate}, line width=1pt] (base.south) -- ($ (cover.north) + (0,0.3em) $);
  \end{tikzpicture}
\end{minipage}
\hspace{3em}
\begin{minipage}{0.45\textwidth}
  \centering
  \begin{tikzpicture}[node distance=1cm, every node/.style={inner sep=0pt}]
    \node (base) {\includegraphics[scale=1.0,page=6]{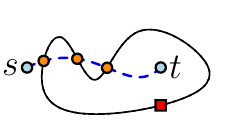}};
    \node (cover) [below=1cm of base] {%
      \fbox{%
        \begin{minipage}{\linewidth}
          \centering
          \includegraphics[scale=1.0,page=7]{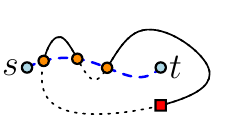}\\
          \frequentedotfill\\
          \includegraphics[scale=1.0,page=8]{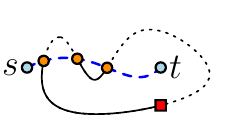}
        \end{minipage}%
      }%
    };
    \draw[|-,decoration={markings,mark=at position 1 with
    {\arrow[scale=2,>=stealth]{>}}},postaction={decorate}, line width=1pt] (base.south) -- ($ (cover.north) + (0,0.3em) $);
  \end{tikzpicture}
\end{minipage}
\caption{Examples of how closed curves in each homology class map into the homology cover,
with a marked ``starting point'': Curves in one homology class map to two distinct closed curves (left),
while the other maps to one longer closed curve (right).}
\label{fig:homology-examples-2}
\end{figure}

\subsection{Homology and Homology Covers Among Obstacles}
\label{subsec:homology-obstacles}

For a set of obstacles $\cC$ with intersection graph $G$
and arrangement $D,\sigma$,
we would like to know which subsets of $\cC$
separate some points $s$ and $t$.

Fix any $s-t$ path $\pi$.
We know from \cref{fact:homology-path-crossings}
that a subset $C\subset\cC$ separates
$s$ and $t$ if and only if it covers a simple closed curve $c$
crossing $\pi$ an odd number of times.
In particular, any such curve $c$ must correspond exactly to a simple cycle $S$ in the arrangement $D$,
and the subset $C\subset\cC$ must correspond exactly to a subset $C'\subset\sigma$.
For a simple $s-t$ path $\pi$ within the dual graph $D^*$,
the homology class of $S$ can be determined by counting the number of edges
of $S$ present in the edges dual to $\pi$, modulo $2$.
This allows us to define a graph-theoretic form of the homology cover w.r.t. $\pi$
specific to the arrangement:
\begin{definition}
For a plane graph $D$
and a simple path $\pi$ through its dual graph,
the \defn{homology cover of $D$ w.r.t. $\pi$}
is a graph $\overline D$ with $2|V(D)|$ vertices
and $2|E(D)|$ edges, defined as follows:
For each vertex $v\in V(D)$,
create two corresponding vertices $v^-,v^+$ in $V(\overline D)$.
For each edge $uv\in E(D)$,
create two corresponding edges:
If $uv$ is dual to an edge in $\pi$,
create the edges $u^-v^+$ and $u^+v^-$.
Else, create the edges $u^-v^-$ and $u^+v^+$.
\end{definition}
We also obtain a simpler case of \cref{lemma:homology-cover-correctness}:
\begin{lemma}
For a plane graph $D$
and a simple dual path $\pi$,
the homology cover $\overline D$ of $D$ w.r.t. $\pi$
has the following property:
There is a path between $v^-,v^+\in V(\overline D)$
passing through the vertex sequence $v^-=u_1^{b_1},u_2^{b_2},\dots,u_k^{b_k}=v^+$,
where $b_i\in\{-,+\}$,
if and only if there is a simple cycle $c=(u_1,u_2,\dots,u_k)$ in $D$
passing through the dual edges of $\pi$
an odd number of times.
\end{lemma}

We will now discuss how to augment the homology cover to consider the obstacles $\cC$
(or more specifically, their representation in the arrangement, $\sigma$).
We will use an ``auxiliary graph''.
\begin{definition}
\label{def:auxiliary-graph}
Let $D,\sigma$ be an arrangement,
and let $\pi$ be a simple dual path between faces $s$ and $t$.
Let $\overline D$ denote the homology cover of $D$.
For each obstacle covering-set $\c'\in\sigma$,
fix a ``canonical'' vertex
$v_\c\in\c'\subset E(D)$.
We define the \defn{auxiliary graph} $H$ as a bipartite graph with vertex set $V(\overline D)\cup(\sigma\times\{-,+\}$.
We call the first part of the vertices the \defn{arrangement vertices}, and the second part the \defn{obstacle vertices}.
Consider a pair of vertices $v^b\in V(\overline D)$ (with $b\in\{-,+\}$) and $(\c',b')\in\sigma\times\{-,+\}$.
If $\c'$ contains an edge incident to $v$ in $D$,
there is a path from $v$ to $v_\c$
crossing the edges dual to $\pi$
an even number of times,
and $b=b'$,
then we add an edge between $v^b$ and $(\c',b')$.
Similarly, we also
add an edge of there is a path crossing edges dual to $\pi$
an odd number of times,
and $b\neq b'$.
Together, these cases are the full set of edges.
\end{definition}
This is somewhat similar to a construction of \cite{KumarLSS21},
but our graph is smaller.
This graph has one very important property:
\begin{lemma}
\label{lemma:auxiliary-graph-useful}
Let $D,\sigma$ be an arrangement,
let $\pi$ be a dual $s-t$ path,
and let $H$ be an auxiliary graph.
Let $P$ be a simple path in $H$
from $v^+$ to $v^-$,
for some vertex $v\in V(D)$,
\emph{or} a simple path from $(\c,-)$
to $(\c,+)$ for some obstacle $\c$.
In either case, the set of obstacles corresponding
to the obstacle vertices along $P$
separates $s$ and $t$.
Moreover, for any set of obstacles $C$ separating $s$ and $t$,
some path of each form exists
including a (possibly equal) subset of $C$.
\end{lemma}

\begin{proof}
In the forward direction:

For the case where we begin with a simple path $P$ from some $(\gamma,-)$ to $(\gamma,+)$,
we will reduce it to the case of a path from some $v^+$ to some $v^-$.
Consider the crossing point (arrangement) vertex $v^+$ encountered in the path (we use $+$ w.l.o.g.).
Then $(\gamma,+)$ must be incident to $v^-$,
and so there is a simple path $P'$ from $v^+$ to $v^-$
inducing the same set of obstacles.

Suppose now that there is a path $P=[v^+=v_1^{s_1},(\gamma_1,b_1),v_2^{s_2},(\gamma_2,b_2),v_3^{s_3},\dots,v^-]$.
Then, by the construction of $H$,
there exists a path from $v_i$ to $v_{\gamma_i}$ crossing $\pi$ $s_i+b_i$ times$\pmod 2$.
Similarly, there exists a path from $v_{\gamma_i}$ to $v_{i+1}$ crossing $\pi$ $b_i+s_{i+1}$ times$\pmod 2$
(we arbitrarily assign $-$ as $1$ and $+$ as $0$ for this purpose).
In total, this path crosses $\pi$ $0+2b_1+2s_2+\cdots+1=1\pmod 2$ times,
so we are done by
\cref{fact:homology-path-crossings},
since all these sub-paths are sub-paths of the corresponding set of obstacles
(and hence the full path is included in their union).\\

In the backwards direction:

Start with some set of obstacles $C$ that separates $s$ and $t$.
Consider the boundary of the region of $\R^2\setminus(\cup_{c\in C}c)$
containing $s$.
This boundary is a subset of $\cup_{\gamma\in C}\gamma$.
In particular, this boundary is a closed curve, and it also separates $s$ and $t$.
Since the curves are closed, the boundary has a finite (circular) sequence of obstacles
it uses.
Pick an arbitrary starting point to get the sequence $\gamma_1,\gamma_2,\dots,\gamma_k$.
These obstacles also have corresponding intersection points, which we label to get the sequence
$v_1,\gamma_1,v_2,\gamma_2,v_3,\dots,\gamma_k$.
The boundary inducing this sequence crosses $\pi$ an odd number of times by
\cref{fact:homology-path-crossings}.
We wish to show that
there is a path
$v_1^{s_1},(\gamma_1^,b_1),v_2^{s_2},(\gamma_2,b_2),v_3^{s_3},\dots,(\gamma_k,b_k),v_1^{s_{k+1}}$
in $H$.
We can actually determine the values $\{s_i\}$ independently from $\{b_i\}$:
Arbitrarily choose $s_1=+$.
Then, pick $s_{i+1}$ to be $s_i$ plus the number of crossings in the path from $v_i$ to $v_{i+1}$ along the boundary of the original region, modulo $2$.
This guarantees that $s_{k+1}=-$, since we know that this boundary crosses $\pi$ and odd number of times.
To pick $b_i$, simply check if there is a path from $v_{\gamma_i}$ to $s_i$ crossing $\pi$ an even number of times.
If so, pick $b_i=s_i$, else there is a path crossing $\pi$ and odd number of times, and we pick $b_i=1-s_i$.
In either case, there is a matching path from $v_{\gamma_i}$ to $v_{i+1}$ crossing $\pi$
$s_{i+1}-s_i$ times (modulo $2$), by taking the combined chosen paths from $v_{\gamma_i}$ to $v_i$ to $v_{i+1}$,
so the path must exist in $H$ as well.
\end{proof}

This will be enough properties of homology for one of our main results.
For the other, we will need an analogous result in a form of intersection graph:
\begin{definition}
Let $D,\sigma$ be an arrangement
of a set of obstacles $\cC$,
let $\pi$ be a dual $s-t$ path,
and let $H$ be an auxiliary graph.
The \defn{intersection graph in the homology cover} is a graph $\overline G$
obtained by replacing each vertex $v^b$
with a set of edges forming a clique of its neighbors.
\end{definition}
Under this definition,
\cref{lemma:auxiliary-graph-useful}
implies that a path between two vertices $(\c,-)$ and $(\c,+)$
through the intersection graph in the homology cover
corresponds to a set of obstacles separating $s$ and $t$,
and vice-versa.

\subsection{Shortest-Path Queries through the Intersection Graph in the Homology Cover}

Let $s,t$ be points in the plane, and let
$\cC$ be a set of (weighted) obstacles given as curves in $\R^2\setminus\{s,t\}$.
The following key lemma characterizes the usefulness of \afterreview{shortest paths} in the homology cover:
\begin{lemma}
\label{lemma:sp}
For a set of obstacles $\cC$,
a weight function $w:\cC\to\R$,
assign weights to all (directed) edges in the auxiliary graph $H$
to be the weight of the obstacle at the end (or $0$ if it is not an obstacle vertex).
Assign weights to all (directed) edges in the intersection graph in the homology cover
in the same manner.
Then
a minimum-weight subset $C$ of the obstacles separating $s$ and $t$
is given by a sequence of obstacles corresponding to vertices
along any shortest path of the form $(\gamma,-),\dots,(\gamma,+)$
for an obstacle $\gamma\in\cC$
through the auxiliary graph or the intersection graph in the homology cover,
\emph{and} by a sequence of obstacles
corresponding to vertices along any shortest path
of the form $v^+,\dots,v^-$
through the auxiliary graph.
\end{lemma}

\begin{proof}
    There is a direct correspondence between these paths and subsets by \cref{lemma:auxiliary-graph-useful}.
\end{proof}

The purpose of this lemma is as follows:
It reduces the problem of finding a \emph{global}
minimum separating set of obstacles $C\subset\cC$
to the minimum of \emph{local} shortest-path problems.
This is analogous to (but not quite the same as) constructions
in each of \cite{KumarLSS21} and \cite{CabelloG16,CabelloM18}.
In particular, \cite{KumarLSS21} also
reduced their formulation of the problem to a set of shortest-path queries
using a carefully constructed graph
(theirs is analogous to a hybrid of our auxiliary graph and our intersection graph in the homology cover),
while \cite{CabelloG16} and \cite{CabelloM18} each used an extra step
beyond shortest-path queries.
In particular, the extra step in \cite{CabelloM18} for unit disks is the bottleneck step,
so we will be able to obtain a simultaneous improvement and simplification of their result.
We are now ready to devise algorithms in all three paradigms,
with varying combinations of weighted/unweighted obstacles.

\section{Static Intersection in the Homology Cover}
\label{sec:static-intersection-hc}

In this section, we give the proof of
\cref{lemma:si-values-homology-cover},
and
a visualization of the algorithm for disks in \cref{fig:disk-algorithm}.

\begin{figure}
    \centering
    \begin{subfigure}[t]{0.47\textwidth}
        \centering
        \includegraphics[page=1,scale=0.45]{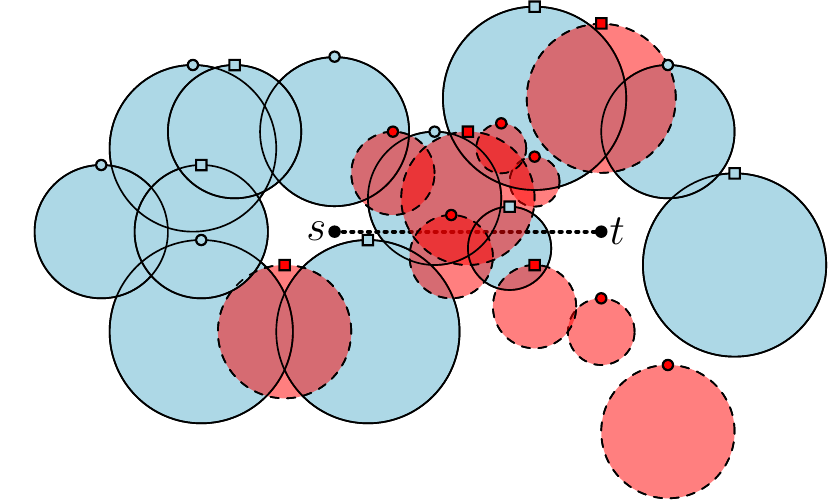}
        \caption{The input to the problem. The set $A$ is given as solid blue disks, and the set $B$ is given as dashed red disks.
        The canonical points are denoted with small filled squares or circles, depending on the indicator bit of the obstacle.}
    \end{subfigure}
    \hfill
    \begin{subfigure}[t]{0.47\textwidth}
        \centering
        \includegraphics[page=2,scale=0.45]{disk-algorithm}
        \caption{The restriction of the problem to one plane copy, given by the ``square indicators''.}
    \end{subfigure}
    \begin{subfigure}[t]{0.47\textwidth}
        \centering
        \includegraphics[page=3,scale=0.45]{disk-algorithm}
        \caption{The restriction of the problem to the other plane copy, given by the ``circle indicators''.}
    \end{subfigure}
    \hfill
    \begin{subfigure}[t]{0.47\textwidth}
        \centering
        \includegraphics[page=5,scale=0.45]{disk-algorithm}
        \caption{The restriction of the problem to the circle indicators plane exclusively above the line $l$
        (the sets $A''$ and $B''$).}
    \end{subfigure}
    \begin{subfigure}[t]{0.47\textwidth}
        \centering
        \includegraphics[page=7,scale=0.45]{disk-algorithm}
        \vspace{-4em}
        \caption{The union of the full disks in $\overline{A''}$.}
    \end{subfigure}
    \hfill
    \begin{subfigure}[t]{0.47\textwidth}
        \centering
        \includegraphics[page=8,scale=0.45]{disk-algorithm}
        \vspace{-4em}
        \caption{The union of the disks and clipped disks in $A''$, obtained by clipping the union of full disks in $\overline{A''}$.}
    \end{subfigure}
    \begin{subfigure}[t]{0.47\textwidth}
        \centering
        \includegraphics[page=10,scale=0.45]{disk-algorithm}
        \vspace{-4em}
        \caption{The queries in $B''$ and their answers via the union of the disks and clipped disks in $A''$.}
    \end{subfigure}
    \caption{A demonstration of the static intersection algorithm for disks in the homology cover.}
    \label{fig:disk-algorithm}
\end{figure}

\begin{proof}
Assume for simplicity that no pair of intersecting pair of obstacles intersects (exclusively) along the line passing through $s$ and $t$ --
since all curves are closed, this can be accomplished by a slight perturbation of $s$ and $t$.
We further assume for simplicity that $\overline{st}$ is along the $x$-axis of the plane.
In all cases, we start with two sets of obstacles $A$ and $B$,
and wish to compute, for each $b\in B$, whether $b$ intersects any element of $A$.
In all cases we will make use of the line segment from $s$ to $t$,
which we denote $\overline{st}$.
We will also make use of ``canonical'' points along each obstacle
(in the sense used by \cref{def:auxiliary-graph}),
which, for simplicity, we assume to be a point with the largest $y$-coordinate in each obstacle
(in some cases, there may be multiple such points, in which case we choose arbitrarily).
Each obstacle in the homology cover also has a parity value (or ``indicator bit'', previously denoted $+$ or $-$)
that essentially indicates which copy of the plane the canonical point
is present in.
For obstacles that do not pass through $\overline{st}$, this indicates
that the entire obstacle is fully contained in that copy of the plane.

For axis-aligned line segments and arbitrary line segments (along with their indicator bits),
we can handle the static intersection problem in the homology cover as follows:
Slice each line segment in each of $A$ and $B$ crossing $\overline{st}$.
That is, clip it both above and below to get two new line segments: One above $\overline{st}$ and one below it.
The one below it needs a new chosen canonical point: Choose the largest $y$-coordinate again,
and flip its indicator bit accordingly.
(inducing new canonical points with flipped indicator bits for the new segment parts below $\overline{st}$).
Call the new sets of segments $A'$ and $B'$:
A segment $b$ in $B$ intersects a segment $a$ in $A$
if and only if one of the (up to two) sub-segments of $b$ in $B'$
intersects one of the (up to two) sub-segments of $a$ in $A'$.
Then, if we cut the homology cover in two spaces according to its ``connection'' along $\overline{st}$,
each resulting segment gets mapped to exactly one copy of the plane.
Equivalently, we partition the resulting segments according to their indicator bits,
since now no sub-segment $b'$ in $B'$ can intersect a sub-segment $a'$ in $A'$
unless they have equal indicator bits.
These two cases can be solved
separately by the algorithms discussed by
Chan and Skrepetos~\cite{ChanS17, ChanS19},
and combined using an \texttt{OR} operation for each element in $B$ that was divided.
Note that solving the static intersection problem for $O(1)$-length polylines
is reduces to solving it for its individual line segments,
so the polyline results follow from the line segment results.

Recall that we have assumed no individual obstacle separates $s$ and $t$,
so no disk contains $s$ or $t$.
Hence, if a disk $D$ intersects $\overline{st}$,
$D\setminus\overline{st}$ also has exactly two connected components
(the top and bottom of the sliced disk).
We will call these ``sliced'' connected components \defn{clipped disks}.
Since we have assumed the intersections along $\overline{st}$ are not unique,
we can also use the closure of these disks.
Thus, following the ideas in the preceding paragraph, we can reduce the static intersection problem with general disks in the homology cover to two static intersection problems with disks and clipped disks \emph{in the plane},
where the clipped disks specifically have their flat boundary along $\overline{st}$.

Call the new sets in this planar static intersection problem $A'$ and $B'$.
Chan and Skrepetos gave an algorithm for solving the static intersection problem with disks in the plane~\cite{ChanS17, ChanS19},
but their algorithm does not immediately extend to clipped disks.
However, their algorithm for disks is quite simple: Create an additively-weighted Voronoi diagram
(equivalently, a Voronoi diagram of disks) for $A'$,
and then use point-location to check for intersections with elements in $B'$.
A natural first idea to extend this is to ask for a Voronoi diagram of disks and clipped disks,
or even just their boundaries.
Unfortunately, the disks in $A'$ are allowed to intersect, and existing algorithms for the Voronoi diagrams of $k$ \emph{disjoint} arcs on the plane take $O(k \log k)$ time~\cite{yap1987n,alt2005voronoi}.
The number of disjoint arcs in the arrangement of $k$ disks can be as high as $O(k^2)$ so we need a different approach for an efficient algorithm.

Our approach is as follows:
Let $l$ be the line through $s$ and $t$ (containing $\overline{st}$).
The sets $A'$ and $B'$ consist of three types of objects:
full disks, clipped disks whose flat boundary is their \emph{top}, along $\overline{st}$,
and clipped disks whose flat boundary is their $\emph{bottom}$, also along $\overline{st}$.
For all the full disks crossed by $l$,
slice them to turn them each into clipped disks whose flat boundaries are along $l$ (but not $\overline{st}$).
To solve the static intersection problem with $A'$ and $B'$,
it suffices to solve it for these further clipped disks.
Moreover, we can now partition all the disks and clipped disks into those \emph{above} $l$ and those \emph{below} $l$,
and handle the static intersection problems in each case \emph{separately}.
By symmetry, we need only devise an algorithm for the clipped disks and full disks lying above $l$.
Call the resulting sets $A''$ and $B''$.

For every disk and clipped disk above $l$ in the query set $B''$, we need to answer if it intersects \textit{some} disk in $A''$.
Let $\overline{A''}$ be the ``extended'' disks. That is, it includes all disks in $A''$, as well as the disks inducing each clipped disk in $A''$.
Our high-level approach will be as follows: First, we will take the union over the elements in $\overline{A''}$, and then clip them using $l$ to get the union over the elements in $A''$.
Finally, we will use point-location for each element in $B''$ with a careful argument to detect intersections.

We now fill in the details. 
The union of the disks in $\overline{A''}$ can be constructed in $O(n \log n)$ time using \textit{power diagrams}~\cite{aurenhammer1988improved}. The boundary of this union is known to have linear complexity~\cite{kedem1986union}.
We can also form the clipped union to retain only the boundary above $l$ in $O(n)$ time
with a traversal through the dual of the arrangement.
The boundary of the clipped union consists of arcs of disk boundaries in $A''$ and line segments along $l$.
With $O(n)$ preprocessing, we can perform point-location in $O(\log n)$ time on the clipped union~\cite{edelsbrunner1986optimal}.
Each element $b\in B''$ is either a disk above $l$, or a clipped disk above $l$ whose lower boundary is $l$ itself.
Denote its centre and radius (or the centre/radius of the disk that induced it if $b$ is a clipped disk) as $c$ and $r$, respectively.
Note that $c$ is inside $b$ if and only if $c$ is above $l$,
which is not always true.
Perform a point-location query from $c$ in the union of elements in $A''$.
If $c$ is inside the union of elements in $A''$
(which are always above $l$)
then clearly $b$ intersects $A''$.
Otherwise, $b$ intersects $A''$ if and only if $c$ has distance $\leq r$ to the union of elements in $A''$, and this distance can also be determined from point-location.
Overall, this algorithm runs in $O(n\log n)$ time.
\end{proof}

\section{Omitted Proofs for Biclique Covers}
\label{sec:appendix-biclique}

In this short section, we give the proofs of \cref{lemma:biclique-cover-to-vw-apsp}
and
\cref{lemma:biclique-covers-homology-cover}.

Below is the proof for \cref{lemma:biclique-cover-to-vw-apsp}:
\begin{proof}
Let the biclique cover of $G$ be denoted $(A_1,B_1),\dots(A_k,b_k)$.
We construct a vertex-weighted directed graph $H=(V\cup U,F)$ whose underlying undirected graph is bipartite, as follows:
\begin{itemize}
    \item The weights of the vertices in $V$ are retained.
    \item Two new weight-$0$ vertices $u_i,u_i'$ are created in $U$ for each biclique $(A_i,B_i)$.
    \item For each biclique $(A_i,B_i)$, and for each $a\in A_i$ create two new (directed) edges in $F$: $(a,u_i)$ and $(u_i',a)$.
    Similarly, for each $b\in B_i$, create edges $(u_i,b)$ and $(b,u_i')$.
\end{itemize}
In this construction, there is a bidirectional correspondence between edges in $G$ and sequential pairs of edges in $H$
centered on a vertex in $U$.
Hence, there is also a correspondence between \afterreview{shortest paths}, and said \afterreview{shortest paths} have identical weighting in both graphs.

Constructing $H$ takes $O(T(n)+S(n))$ time,
and then APSP on $H$ can be computed in $n\cdot S(n)\log(n)$ time.
\end{proof}

Below is the proof for
\cref{lemma:biclique-covers-homology-cover}:

\begin{proof}
We can apply the same reduction to the planar case of line segments as used in the proof of
\cref{lemma:si-values-homology-cover}.
Now, the intersection graph in the homology cover
is the union of two intersection graphs of line segments $O(n)$ in the plane.
In particular, each of the line segments in these intersection graphs is a sub-segment of some segment in $\cC$,
and the resulting intersection graphs are both subgraphs of $G$.
We can obtain biclique covers for each of these graphs, and take the union of the covers to obtain a cover of $\overline G$ of size $\widetilde{O}(n^{4/3})$ in the same time complexity.
If $G$ also has no $k$-clique, then the union of the two intersection graphs has no $2k$-clique
(and hence neither does $\overline G$),
so we can obtain a biclique cover of size $\widetilde{O}_k(n)$ in the same time complexity.
\end{proof}

\section{Lower Bound Details}
\label{sec:lb-details}

In this section, we discuss some background on fine-grained lower bounds,
and state their applications to $(s,t)$ point-separation that follow from
\cref{thm:k-cycle-reduction}.

\subsection{The APSP conjecture}

In \cref{subsec:apsp} we discussed APSP-hardness.
Among the known APSP-hard problems is the \defn{minimum-weight triangle problem}~\cite{WilliamsW10,WilliamsW18}
which asks for a minimum-weight triangle in an edge-weighted undirected $n$-vertex graph.
With this information, \cref{thm:k-cycle-reduction} implies the following results:
\begin{corollary}
\label{cor:lb-segments-apsp}
If the weighted $(s,t)$-point separation problem with $m$ line-segment obstacles
(or length-$3$-rectilinear-polylines)
can be solved in $O(m^{\frac32-\varepsilon})$ time
for some $\varepsilon>0$,
then there exists some $\varepsilon'>0$ such that
edge-weighted APSP on $n$ vertices
can be solved in $O(n^{3-\varepsilon'})$ time.
\end{corollary}

\begin{corollary}
\label{cor:lb-polylines}
If the weighted $(s,t)$-point separation problem with $N$ length-$2$-polyline obstacles
can be solved in $O(N^{\frac32-\varepsilon})$ time
for some $\varepsilon>0$,
even just for the case that the number of unique intersection points of the obstacles is $O\left(\sqrt{N}\right)$,
then there exists some $\varepsilon'>0$ such that
edge-weighted APSP on $n$ vertices
can be solved in $O(n^{3-\varepsilon'})$ time.
\end{corollary}

This second case is particularly interesting in the context of our algorithm
for the arrangement model (\cref{thm:not-match-cg16-arrangement}), where it essentially matches.

\subsection{The $(2l+1)$-Clique Conjecture}

Under a different fine grained lower bound hypothesis,
we can obtain a stronger lower bound.
The $(2l+1)$-clique hypothesis
conjectures that
there is no $O(n^{2l+1-\varepsilon})$-time algorithm
for the minimum-weight $(2l+1)$-clique problem,
for any fixed $\varepsilon>0$.
Lincoln, Williams, and Williams
show that, if this hypothesis holds,
then there does not exist an $O(n^2+mn^{1-\epsilon})$-time
algorithm for shortest $(2l+1)$-cycle
in a directed graph with $m=\Theta\left(n^{1+\frac1l}\right)$ edges~\cite{lincoln2018tight}.
Importantly, their construction operates on a very specific form of graph in which the $k$-cycles coincide exactly with the $k$-walks (moreover, there are no cycles with fewer than $k$ vertices),
called a $k$-circle-layered graph
(as we mentioned before).
We omit the precise definition of this type since the coincidence of $k$-cycles and $k$-walks
is the only property we need for these graphs.

We will use this hypothesis and result here to present a stronger lower bound for
the $(s,t)$ point-separation problem with line segments:

\begin{theorem}
\label{thm:lb-segments-clique}
Assume the $(2l+1)$-clique hypothesis holds.
Then there does not exist
an $O(N^{2-\varepsilon})$-time algorithm
for the weighted $(s,t)$ point-separation problem
with $N$ line segment obstacles (or length-$3$-rectilinear-polyline obstacles),
for any fixed $\varepsilon>0$.
\end{theorem}

\begin{proof}
We give a proof by contradiction.
Assume that $l$ is a constant.
Assume there exists an $O(N^{2-\varepsilon})$-time
algorithm for the $(s,t)$ point-separation
problem with $N$ line segment obstacles (or length-$3$-rectilinear-polyline obstacles)
and some fixed $\varepsilon>0$.
By \cref{thm:k-cycle-reduction},
there is an algorithm for minimum-weight $(2l+1)$-walk (over $m$ edges)
running in $O(m^{2-\varepsilon})$ time.
Pick some fixed $l\geq\frac2\varepsilon$,
and consider an instance of the minimum-weight $(2l+1)$-cycle problem
in a directed graph with $m=\Theta(n^{1+\frac1l})$ edges,
in a graph where the $(2l+1)$-cycles coincide exactly with the $(2l+1)$-walks.
Since $m=\Theta(n^{1+\frac1l})$,
we have an
$O\left(mn^{\left(1+\frac1l\right)(1-\varepsilon)}\right)$-time algorithm for
the min-weight $(2l+1)$-cycle problem in this graph.
Further simplifying:
$$\left(1+\frac1l\right)\left(1-\varepsilon\right)
\leq
\left(1+\frac1l\right)\left(1-\frac2l\right)
=
1+\frac1l-\frac2l-\frac2{l^2}
<
1-\frac1l,
$$
so
$$O\left(mn^{\left(1+\frac1l\right)(1-\varepsilon)}\right)
\subset
O\left(mn^{1-\frac1l}\right).$$
This contradicts the $(2l+1)$-clique hypothesis,
so we are done.
\end{proof}

Since the construction here is essentially the same as that of the lower bound via the APSP conjecture,
the same corollaries also extend:

\subsection{Unweighted Lower Bound}

Lastly, we can apply the same technique to reduce from directed $k$-walk
\emph{detection} (i.e. unweighted)
to the unweighted $(s,t)$ point-separation problem to arrive at the following theorem,
when combined with another result of 
Lincoln, Williams, and Williams~\cite{lincoln2018tight}:

\begin{theorem}
\label{thm:unweighted-lb}
If there exists $\varepsilon>0$ such that unweighted $(s,t)$ point-separation
can be solved in $O(N^{3/2-\varepsilon})$ time
for $N$ obstacles that are uniformly all line segments or length-$3$ rectilinear polylines,
then there exists $\varepsilon'>0$ such that max-$3$-SAT can be solved in $O(2^{(1-\varepsilon')n})$ time for $n$ variables.
\end{theorem}
Note that there are at most $2^3\cdot\binom{n}{3}$ clauses in max-$3$-SAT, and polynomial factors in $n$ can be ``hidden'' by slightly decreasing $\varepsilon'$.
To the best of our knowledge, it is still true that no algorithm for max-$3$-SAT with this time complexity is known in the general case, so progress beyond this threshold on $(s,t)$ point-separation would imply progress on a much more fundamental problem.

\begin{proof}
Assume $k$ is a constant.
Apply \cref{thm:k-cycle-reduction}
to obtain a reduction to unweighted $(s,t)$ point-separation,
If we assume the existence
of an algorithm running in $O(N^{3/2-\varepsilon})$
time for the $(s,t)$ point-separation problem,
we obtain an algorithm for directed $k$-walk over $m$ edges
running in $O(m^{3/2-\varepsilon})$ time.
Lincoln, Williams, and Williams~\cite[Corollary 9.3]{lincoln2018tight}
show that such an algorithm for directed $k$-cycle in a
graph where $k$-walks and $k$-cycles coincide
also implies the algorithm for max-$3$-SAT
with the stated time complexity.
\end{proof}

}

\end{document}